\documentclass[9pt, a4paper]{article}
\usepackage{amsmath,amssymb,amsthm}
\usepackage{fullpage}
\usepackage{fancyhdr}

\usepackage{thm-restate}
\DeclareMathOperator{\Var}{\operatorname{Var}}
\DeclareMathOperator{\E}{\operatorname{E}}

\fancypagestyle{firststyle}
{
\fancyhead{}

\fancyfoot[L]{\begin{minipage}{0.5\textwidth}\footnotesize{\textcopyright Copyright is held by the author. \\Publication rights licensed to ACM.\\PODS '21, June 20--25, 2021,\\ Virtual Event, China}\end{minipage}}
}
\AtBeginDocument{%
  \providecommand\BibTeX{{%
    \normalfont B\kern-0.5em{\scshape i\kern-0.25em b}\kern-0.8em\TeX}}}

\title{\Large{Improved Differentially Private Euclidean Distance Approximation}}
\author{Nina Mesing Stausholm\\\small{nimn@itu.dk}\\\small{IT University of Copenhagen}\\\small{BARC}}
\date{}

\begin{document}
\maketitle
\thispagestyle{firststyle}



\newtheorem{note}{Note}
\newtheorem{theorem}{Theorem}
\newtheorem{lemma}{Lemma}
\newtheorem{definition}{Definition}
\newtheorem{corollary}{Corollary}
\newcommand{\sign}{\varphi} %

\begin{abstract}
This work shows how to privately and more accurately estimate Euclidean distance between pairs of vectors. Input vectors $x$ and $y$ are mapped to differentially private sketches $x'$ and $y'$, from which one can estimate the distance between $x$ and $y$. Our estimator relies on the Sparser Johnson-Lindenstrauss constructions by Kane \& Nelson (Journal of the ACM 2014), which for any $0<\alpha,\beta<1/2$ have optimal output dimension $k=\Theta(\alpha^{-2}\log(1/\beta))$ and sparsity $s=O(\alpha^{-1}\log(1/\beta))$. We combine the constructions of Kane \& Nelson with either the Laplace or the Gaussian mechanism from the differential privacy literature, depending on the privacy parameters $\varepsilon$ and $\delta$. We also suggest a differentially private version of Fast Johnson-Lindenstrauss Transform (FJLT) by Ailon \& Chazelle (SIAM Journal of Computing 2009) which offers a tradeoff in speed for variance for certain parameters.
We answer an open question by Kenthapadi et al.~(Journal of Privacy and Confidentiality 2013) by analyzing the privacy and utility guarantees of an estimator for Euclidean distance, relying on Laplacian rather than Gaussian noise. We prove that the Laplace mechanism yields lower variance than the Gaussian mechanism whenever $\delta<\beta^{O(1/\alpha)}$.
Thus, our work poses an improvement over the work of Kenthapadi et al.~by giving a more efficient estimator with lower variance for sufficiently small $\delta$. Our sketch also achieves \emph{pure} differential privacy as a neat side-effect of the Laplace mechanism rather than the \emph{approximate} differential privacy guarantee of the Gaussian mechanism, which may not be sufficiently strong for some settings.

Our main result is a special case of more general, technical results proving that one can generally construct unbiased estimators for Euclidean distance with a high level of utility even under the constraint of differential privacy. The bulk of our analysis is proving that the variance of the estimator does not suffer too much in the presence of differential privacy. 
\end{abstract}

\section{Introduction}
The well-known Johnson-Lindenstrauss Lemma \cite{johnson1984extensions} is a fundamental tool in dimensionality reduction and has applications in a variety of fields. The lemma allows for significant speed-ups in applications such as nearest-neighbor search \cite{AilonC09,IndykM98}, 
 computational geometry \cite{Clarkson08}, document comparison \cite{TanSK2005},
data streams \cite{Indyk06},
clustering \cite{BoutsidisZMD15,CohenEMMP15}, graph sparsification \cite{SpielmanS11}, low-rank approximation \cite{ClarksonW13},
numerical linear algebra \cite{Woodruff14,ClarksonW09,DrineasMMS11} and many more.
The Johnson-Lindenstrauss Lemma, or JL lemma for short, states that for any $0<\alpha,\beta<1/2$ and input dimension $d>0$, there exists a random $k\times d$-projection matrix $S$ such that $S$ preserves Euclidean norm of any input vector $x\in\mathbb{R}^d$ up to a factor $(1\pm\alpha)$ with probability at least $1-\beta$. Classic examples of projections satisfying the lemma include the constructions from \cite{IndykM98,AilonC09, Achlioptas03,DasguptaKS10,KaneN14}. Jayram \& Nelson \cite{JayramW11}, and later Kane et al.
~\cite{KaneMN11}, proved the remarkable result that the optimal output dimension $k$ is independent of the input dimension $d$. In particular, they showed that $k=\Theta(\alpha^{-2}\log(1/\beta))$ is optimal.
In the case where $n$ input vectors are known in advance, the \emph{Johnson-Lindenstrauss Flattening Lemma} states that there exist projections preserving Euclidean distance for \emph{all} pairs of these vectors within a factor $(1\pm \alpha)$, simultaneously. We concern ourselves with a distributed setting, where data is held by several parties and may not be present at the same time. Hence, in our setting the length preserving projection $S$ must be public, so any party holding input vector $x\in\mathbb{R}^d$ can compute and release $Sx$. For inputs $x,y\in\mathbb{R}^d$ held by different parties, one can estimate the Euclidean distance as $\Vert Sx-Sy\Vert_2^2~=~\Vert S(x-y)\Vert_2^2$. By the Johnson-Lindenstrauss Lemma, this estimate is within a factor $(1\pm\alpha)$ of $\Vert x-y\Vert_2^2$ with high probability. We henceforth use \emph{transform} and \emph{projection} interchangeably and refer to random projections satisfying the Johnson-Lindenstrauss Lemma as \emph{Johnson-Lindenstrauss projections}, or simply \emph{JL projections}. We will even misuse this convention slightly, as we also use this name for projections that preserve Euclidean norm in \emph{expectation}, as defined in Definition \ref{def:lengthpreservingproperty}.

As the input $x$ may contain sensitive information, the released projection of $x$ must preserve \emph{privacy} to prevent third parties from learning the input $x$. Privacy has often been obtained through simple anonymization by removing obvious identifiers, but several cases have shown that this approach is insufficient \cite{anonymization,de2013unique,NarayananS08,sweeney2015only}. Due to its stringent definition and provable guarantees, we concern ourselves with \emph{differential privacy} \cite{DworkMNS06}, which is usually achieved by perturbing the result of a query, to obfuscate the true result slightly. Thus, we analyze the privacy and utility guarantees when adding noise to the projection $Sx$. That is, for a noise distribution $\mathcal{D}$ and noise vectors $\eta,\mu\in\mathcal{D}^k$, we analyze whether we can privately and accurately estimate $\Vert x-y\Vert_2^2$ from $Sx+\eta$ and $Sy+\mu$. The main questions of interest are: \emph{How much noise do we need to add?} and \emph{What utility guarantees can we achieve?} We define differential privacy and mention common choices for noise distribution $\mathcal{D}$ in Section \ref{sec:DP}.

\subsection{Differentially Private Random Projections}
\label{sec:Kenthapadi}
This work improves on the work of Kenthapadi et al.~\cite{KenthapadiKMM13}, in which it was shown how to construct an $(\varepsilon,\delta)$-differentially private version of a JL transform allowing for high accuracy estimators for squared Euclidean distance. The idea applied by Kenthapadi et al.~is simple: Let $P$ be the i.i.d. normally distributed JL transform where each entry is drawn from the standard Normal distribution. For input vector $x\in[0,1]^d$, add Gaussian noise to each entry of $Px$.

Kenthapadi et al.~prove Theorems \ref{thm:kethapadiSigmaPrivacy} and \ref{thm:kenthapadiUnbiasedEstVar}. We remark that these results extend naturally to $x,y\in \mathbb{R}^d$.
\begin{theorem}[\cite{KenthapadiKMM13}]
\label{thm:kethapadiSigmaPrivacy}
Let $P$ be a $k\times d$-projection matrix with i.i.d. entries from the standard Normal distribution

and let $x,y\in[0,1]^{d}$ be input vectors. Let $\eta,\mu\sim \mathcal{N}(0,\sigma^2)^{k}$ be noise vectors. If $\sigma~\ge~4/\varepsilon\sqrt{\log(1/\delta)}$, $\varepsilon<\ln(1/\delta)$ and $k>2(\ln(d)+\ln(2/\delta))$, then $Px+\eta$ is $(\varepsilon,\delta)$-differentially private.
\end{theorem}
\begin{note}
\label{note:kenthapadinoteone}
Kenthapadi et al.~show that for $k>2\ln(d)+2\ln(1/\delta')$, the $\ell_2$-sensitivity of $P$ is greater than 2 with probability at most $\delta'$. We will assume that the $\ell_2$-sensitivity of $P$ is computed exactly in an initializing step, as discussed in Section \ref{sec:relatedworkKenthapadi}, and hence avoid this assumption on $k$. From \cite{JayramW11,KaneMN11} we know that for any $\alpha,\beta\in(0,1/2)$ $k=\Theta\left(\alpha
^{-2}\log(1/\beta)\right)$ is optimal in the non-private case. We use this value of $k$ and discuss the optimal $k$ for the noisy construction in Section \ref{sec:optimalknoisy}. We also remark that the $\sigma$ in Theorem \ref{thm:kethapadiSigmaPrivacy} can be exchanged with $\sigma \ge \Delta_2\varepsilon^{-1}\sqrt{2\log(1.25/\delta)}$ by a later result from \cite{DworkR14} (See Lemma \ref{lem:GaussMech}), where $\Delta_2$ is the (exact) $\ell_2$-sensitivity of $P$.
\end{note}

\begin{theorem}[\cite{KenthapadiKMM13}]
\label{thm:kenthapadiUnbiasedEstVar}
Let $P$ be a $k\times d$-projection matrix with i.i.d. entries from the standard Normal distribution and $x,y\in[0,1]^{d}$ be input vectors. Let $\eta,\mu\sim \mathcal{N}(0,\sigma^2)^{k}$ be noise vectors, where $\sigma$ is independent of the realization of $P$. Define
\[\hat{E}_{iid}:=\Vert (Px+\eta)-(Py+\mu)\Vert_2^2-2k\sigma^2.\] Then
\begin{enumerate}
    \item $\hat{E}_{iid}$ is an unbiased estimator for $\Vert x-y\Vert_2^2$.
    \item $\operatorname{Var}\left[\hat{E}_{iid}\right]=\frac{2}{k}\Vert x-y\Vert_2^4+8\sigma^2\Vert x-y\Vert_2^2+8\sigma^4k.$
\end{enumerate}
\end{theorem}
\begin{note}
\label{note:sigmaKenthapadi}
Letting $\sigma$ be independent of the realization of $P$ might lead to complete loss of privacy if the $\ell_2$-sensitivity of $P$ is much higher than 1, as argued in Section \ref{sec:relatedworkKenthapadi}. Hence, we let $\sigma$ be a function of the exact $\ell_2$-sensitivity of $P$, $\Delta_2$, as discussed in Note \ref{note:kenthapadinoteone}.
\end{note}

\subsection{New Contributions}
\label{sec:contributions}
An immediate idea to achieve a speed-up is to apply the techniques of Kenthapadi et al.~to a JL transform, which is faster than the i.i.d. normally distributed JL transform. We show such a result for a private Fast Johnson-Lindenstrauss Transform (FJLT) \cite{AilonC09} in Section \ref{sec:privateFJLT}, but remark that the privacy issue of Kenthapadi et al.~mentioned in Note \ref{note:sigmaKenthapadi} carries over, if we simply exchange the i.i.d. normally distributed JL transform for the FJLT. We discuss how to address this issue in Section \ref{sec:FJLT} to obtain a differentially private version of FJLT, where $\sigma$ does not depend on the $\ell_2$-sensitivity of the transform (which could be very large).
Kenthapadi et al.~leave open the question of whether we can obtain better results with Laplacian noise. We answer this question by proving that we can indeed obtain an $\varepsilon$-differentially private estimator for squared Euclidean distances, which has better variance for certain parameters. Specifically, we show the following main theorem:
\begin{restatable}{theorem}{main}
\label{thm:main}
For any $0<\alpha,\beta<1/2$ and any integer $d>0$ there exists a random $k\times d$-projection $S$ for $k=\Theta\left(\alpha^{-2}\log(1/\beta)\right)$ with sparsity $s=O\left(\alpha^{-1}\log(1/\beta)\right)$ and a distribution $\mathcal{D}$ over $\mathbb{R}$ such that for any $x,y\in\mathbb{R}^d$ and $\eta,\mu\sim\mathcal{D}^k$ we define:
\[
\hat{E}_{SJLT}:=\Vert (Sx+\eta)-(Sy+\mu)\Vert_2^2-\frac{2ks}{\varepsilon^2}.
\]

Then
\begin{enumerate}
    \item $\hat{E}_{SJLT}$ is an unbiased estimator for $\Vert x-y\Vert_2^2$.
    \item $\operatorname{Var}\left[\hat{E}_{SJLT}\right]\le \frac{2}{k}\Vert x-y\Vert_2^4+O\left(\frac{s}{\varepsilon^2}\Vert x-y\Vert_2^2+\frac{s^2}{\varepsilon^4}k\right).$
    \item The sketch $(S,Sx+\eta)$ is $\varepsilon$-differentially private.
    \item For a data stream, we can update the sketch $(S,Sx+\eta)$ in time $O(s)$. 
    \item $Sx+\eta$ can be computed in time $O(s\Vert x\Vert_0+k)$. Given $Sx+\eta$ and $Sy+\mu$, $\hat{E}_{SJLT}$ can be computed in time $O\left(k\right)$.
\end{enumerate}
\end{restatable}
\medskip
The noise distribution $\mathcal{D}$ will depend on the sparsity of $S$ but it is crucial that $\mathcal{D}$ is otherwise independent of $S$.
We state our improvements over the work of Kenthapadi et al.~\cite{KenthapadiKMM13}:
\begin{itemize}
    \item Recall that the projection $P$ of Kenthapadi et al.~has constant $\ell_2$-sensitivity with high probability. Under this assumption, we combine Theorems \ref{thm:kethapadiSigmaPrivacy}, \ref{thm:kenthapadiUnbiasedEstVar} and Note \ref{note:kenthapadinoteone} to see that
    \[
    \operatorname{Var}[\hat{E}_{iid}]=\frac{2}{k}\Vert x-y\Vert_2^4+O\left(\frac{\log(1/\delta)}{\varepsilon^2}\Vert x-y\Vert_2^2+\frac{\log^2(1/\delta)}{\varepsilon^4}k\right),
    \] and so $\hat{E}_{SJLT}$ improves over $\hat{E}_{iid}$ in terms of variance whenever $\delta<e^{-s}=\beta^{O(1/\alpha)}$ (see Section \ref{sec:compareFJLTSparser}). In the case where $P$ has higher sensitivity, our results give an even better improvement.
    \item Kenthapadi et al.~have an additional initialization cost of $O(dk)$ to compute the sensitivity of the projection matrix. We refer the reader to Section \ref{sec:relatedworkKenthapadi} for a detailed discussion.
    \item Our estimator $\hat{E}_{SJLT}$ is more efficient as the update time, i.e., time to compute $Sx+\eta$, is $O(s\Vert x\Vert_0~+~k)$ rather than $O(k\Vert x\Vert_0+k)$ for $s=o(k)$.
    \item Rather than \emph{approximate} differential privacy, which may be insufficient for some applications, we achieve \emph{pure} differential privacy.
\end{itemize} 
Our improved efficiency in Theorem \ref{thm:main} relies on the sparsity of the Sparser JL transforms by Kane \& Nelson \cite{KaneN14}, henceforth referred to as \emph{the SJLT}. We remark that the results of Kenthapadi et al.~extend naturally to these JL transforms, and thus they would obtain the same efficiency for $\delta>\beta^{O(1/\alpha)}$. We do, although, give the analysis proving that these transforms can indeed be used. Using a SJLT instead of the i.i.d. normally distributed transform, the work of Kenthapadi et al.~would also avoid the initialization cost.

Related to our analysis for the SJLT, we remark that our main result is, in fact, a special case of an even more general result: we give a class of length preserving linear transformations that allow for efficient, private estimators for Euclidean distance with a high level of utility. The FJLT and SJLT are merely examples of such linear transformations. We define what is meant by \emph{length preserving} in section \ref{sec:lengthpreservingproperty} and prove our general, technical results in Section \ref{sec:generalization}. In Section \ref{sec:FJLT} we give two differentially private versions of FJLT
and in Section \ref{sec:applications}, we prove Theorem \ref{thm:main} by applying the technical results to the SJLT with noise from the Laplace distribution. Finally, we compare the work of Kenthapadi et al.~with our private FJLT and SJLT in Section \ref{sec:compareFJLTSparser}.

\section{Related Work}

Differential privacy is usually achieved by adding random noise to the output of a query to obfuscate the exact result, before publishing the result. This idea is easily extended to vector outputs by simply adding noise to each entry of the output vector. This technique has been studied extensively in previous work, see for example \cite{MirMNW11,PaghStausholm, HsuKR12, McSherryM09}. 

We consider a distributed setting, where party $i$ adds noise $\eta_i\sim\mathcal{D}^k$ to the projection $Sx_i$ of input vector $x_i$ and releases the noisy projection $Sx_i+\eta_i$ for future distance estimation. All parties must use the same randomized matrix $S$ and noise drawn from the same distribution $\mathcal{D}$. It is crucial that the projection matrix is public, and only the noise be kept secret.

\subsection{Versions of Johnson-Lindenstrauss Transformations}
\label{sec:versionsJLT}
We refer to the classical JL transform by Indyk \& Motwani \cite{IndykM98} as the \emph{i.i.d. normally distributed JL transform}. As the name suggests, the random projection matrix consists of i.i.d. entries from the standard Normal distribution.

The \emph{sparsity} of the random projection, i.e., the number of non-zero entries per column is an important tool in speeding up dimensionality reduction. Ailon \& Chazelle \cite{AilonC09} presented a JL transform with a sparser projection matrix with a mixture of normally distributed entries and 0s. This transform is commonly known as \emph{The Fast Johnson-Lindenstrauss Transform} or in short, \emph{FJLT}. We describe the transform in detail in Section \ref{sec:FJLTDescription}.

The sparsity not only affects the sensitivity of the transformation (see Section \ref{sec:sensitivitydef} for the definition of sensitivity), but also the time required to compute the projection of an input vector $x$. For a random projection $S$ with sparsity $s$, we can compute $Sx$ in time $O(s\Vert x\Vert_0)$.
Kane \& Nelson \cite{KaneN14} show that the JL transform of Dasgupta et al.~\cite{DasguptaKS10} requires sparsity $s=\tilde{\Omega}(\alpha^{-1}\log^2(1/\beta))$, and Nelson \& Nguyen showed that this sparsity is optimal up to a factor $O(\log(1/\alpha))$ \cite{NelsonN13}. 
Kane \& Nelson \cite{KaneN14} also give two sparser constructions with $s=\Theta(\alpha^{-1}\log(1/\beta))$ for embedding into $k=\Theta(\alpha^{-2}\log(1/\beta))$ dimensions. These transformations are commonly known as \emph{The Sparser JL Transforms} and we will henceforth refer to them as \emph{SJLT}. We describe SJLT in Section \ref{sec:Kaneconstruction}.

\subsubsection{Differentially Private JL Construction}
\label{sec:relatedworkKenthapadi}

Kenthapadi et al.~\cite{KenthapadiKMM13}, which was also discussed in Section \ref{sec:Kenthapadi}, give a private estimator for Euclidean distance relying on the i.i.d. normally distributed JL transform. 
A drawback of their construction is that the $\ell_2$-sensitivity is only 1 in \emph{expectation}, so the sensitivity \emph{might} not be small. This is the case if the random projection has even a single very large entry. The authors suggest drawing noise calibrated to a low sensitivity projection matrix independently of the actual projection matrix $P$. However, with a small probability, $P$ \emph{does not} have low sensitivity, in which case the noise is not ensured to provide differential privacy. Kenthapadi et al.~''hide'' the probability of drawing a high-sensitivity projection under $\delta$, but for a fixed $P$, either the noise provides privacy, or certain inputs would always be distinguishable, even in the presence of noise calibrated to low sensitivity. 
An alternative solution is to compute the sensitivity of the fixed $P$ and calibrate the noise to the actual sensitivity. Hence, initialization requires time $O(dk)$. 
Kenthapadi et al.~state without proof that their results extend to the JL transformations from \cite{Achlioptas03,DasguptaKS10}.
Xu et al.~\cite{XuRZQR17} extend the work of \cite{KenthapadiKMM13} with experimental comparisons with JTree \cite{ChenXZX15}, PrivBayes \cite{ZhangCPSX14}, PriView \cite{QardajiYL14} and PrivateSVM \cite{RubinsteinBHT12}.

\subsection{Differentially Private Linear Transformations}
\label{sec:relatedworkMir}
Mir et al. (PODS11)~\cite{MirMNW11} suggest a general framework for generating pan-private linear transformations by initializing with noise from the exponential mechanism. The work argues how to create a $\varepsilon$-pan private estimator for (squared) Euclidean distance with multiplicative error $(1+\gamma)$ and additive error $\operatorname{poly}(\log d, \varepsilon^{-1},\gamma^{-1},\log(q^{-1}))+O(Z)$, with probability at least $1-q$, where $Z$ is an upper bound on the entries of the input vector. 
The technique used by Mir et al.~can be used for private dimensionality reduction, but is computationally inefficient as the sketch relies on the exponential mechanism for noise addition.

In an earlier (unpublished) version of the same work, \cite{MirImprovedArXiv}, Mir et al.~analyze the \emph{cropped} second moment for a parameter $\tau$, defined for input vector $x\in\mathbb{Z}^d$ as $\sum_{i=1}^d\min\{x_i^2,\tau\}$. In this work, Mir et al.~show a $2\varepsilon$-differentially private estimator with additive error $O_{\varepsilon}(\tau\sqrt{d})$ with high probability. Differential privacy is achieved by an application of Randomized Response \cite{RandomizedResponse}. As our error depends on $\Vert x-y\Vert_2$ and $\sqrt{k}<\sqrt{d}$, we see an improvement when $x$ and $y$ are sparse. The problems are not directly comparable as the cropped second moment of Mir et al.~applies to integer inputs, whereas we consider inputs over the reals.

\subsection{When Data is Known in Advance}
If input data is known in advance, there are other techniques to achieve differential privacy. A central unit with access to all data can compute the \emph{exact} distances (up to the error incurred by the JL embedding) and add noise specifically calibrated to this distance. This technique often incurs less noise, but is not applicable in our setting, as data is split among several parties and may not all be available at once. 

Blocki et al.~\cite{BlockiBDS12} show that, as long as the projection matrix is kept secret, the i.i.d. normally distributed JL transform 
allows for differentially private estimates of distances with the accuracy guarantees from the Johnson-Lindenstrauss Lemma. 
Upadhyay \cite{upadhyay2014randomness} proves that this technique does not generally work to preserve privacy for sparser JL projections. As we consider a distributed setting, keeping the projection matrix secret is unattainable.
Bhaskar et al.~\cite{BhaskarBGLT11} introduce \emph{noiseless privacy} where the output is always exact, rather than a noisy approximation. The privacy guarantees are of a similar form as differential privacy but rely on assumptions about the distribution of the data and auxiliary information, whereas differential privacy aims for a higher level of generality.

\subsubsection{Representing Noise from Continuous Noise Distributions}
We will assume that noise is drawn from either the continuous Laplace or Gaussian distribution, which, however, may introduce practical issues.
Mironov \cite{Mironov12} described how privacy may be lost due to floating-point error when sampling noise from a continuous distribution. As an alternative to the continuous Laplace distribution, Mironov suggests the \emph{Snapping mechanism}, which incurs an additional error of approximately $\Delta_1/\varepsilon$ compared to noise from $\operatorname{Lap}(\Delta_1/\varepsilon)$, where $\Delta_1$ is the $\ell_1$-sensitivity of the query. 

\cite{securenoisegeneration} improve over the Snapping Mechanism, by drawing noise from a discrete distribution, differing from the Laplace distribution by at most a factor $(1+\frac{1+2/\varepsilon}{2^k})$ for a fixed integer $k$, which controls the accuracy of the discretization. It suffices to use $k\in[10,45]_{\mathbb{Z}}$. 

A discrete, ''hole-free'' alternative to the Gaussian distribution, requiring only expected constant time is suggested in \cite{securenoisegeneration}. The distribution builds on the Binomial distribution with parameters $n$ and $p=1/2$ and the work of \cite{BringmannKPPT14} to give a distribution which for large $n$ differs from the Gaussian distribution by at most $O(\log^{1.5}(n)/\sqrt{n})$. 
 
 In a very recent work, Canonne et al.~\cite{CanonneKS20} describe a discretization of the Gaussian distribution supported on $\mathbb{Z}$ whose variance is at most that of the corresponding continuous Gaussian distribution, and hence allows for identical or slightly better utility. Simultaneously, the discretization has sub-Gaussian tails compared to the corresponding continuous Gaussian distribution and essentially the same privacy guarantees.  We refer to the discussion in \cite{CanonneKS20} for further reading on discretizations of the Laplace and Gaussian distributions.

\subsection{Lower bounds}
McGregor et al.~\cite{mcgregor2010limits} show that any protocol for estimating Hamming distance (and so for inner product, which again leads to a protocol for estimating squared Euclidean distance)
of two binary $k$-dimensional vectors in a differentially private manner incurs an additive error of $\tilde{\Omega}(\sqrt{k})$, which is contrasted by the observation that simple Randomized Response \cite{RandomizedResponse} allows for error $O(\sqrt{k})$. The error lower bound implies a $\tilde{\Omega}(k)$ lower bound for the variance of the noisy estimator. In contrast, our  variance of the noise added (we may disregard the variance introduced by the JL projection, as this error occurs even in the non-private version) depends on $\Vert x-y\Vert_2^2\le d$ and $k$ (for binary input vectors). 

Independently from the work of McGregor et al., Mir et al.~\cite{MirMNW11} show a similar lower bound of additive error $\Omega(\sqrt{k})$ for estimating inner product for binary vectors in a pan-private setting. The lower bound by McGregor et al.~implies a lower bound for pan-private algorithms, which is weaker than the lower bound of Mir et al.~in the case of single-pass algorithms and dynamic data.
Hardt \& Talwar \cite{HardtT10} show that an $\varepsilon$-differentially private algorithm for the second frequency moment $F_2$ requires an additive error factor of $\Omega(1/\varepsilon)$, which is comparable to our result (up to polynomial and logarithmic factors).

\section{Preliminaries}
\subsection{Notation}
Let $x\in\mathbb{R}^{d}$ be an input vector. For a $k\times d$-matrix $S$, let $Sx$ be the linear transformation of $x$ under $S$. Let $\eta\sim\mathcal{D}^k$ for a noise distribution $\mathcal{D}$. We denote by $Sx+\eta$ the noisy counter-part to $Sx$ and let $\eta_*\sim\mathcal{D}$ denote a random variable drawn according to $\mathcal{D}$.
We use \emph{transformation of $x$} and \emph{projection of $x$} interchangeably as our main focus will be on random projections. Unless otherwise specified \emph{projection} always refers to a \emph{random projection}.

Denote by $\Vert x\Vert_p$ the $\ell_p$-norm of $x$ and let $1[p]$ denote the indicator variable for predicate $p$.

\subsection{Differential Privacy}
\label{sec:DP}
Intuitively, differential privacy guarantees that one cannot (confidently) distinguish between whether an output is the result generated from a specific input vector or from a \emph{neighboring} vector:
\begin{definition}[Neighboring inputs]
Vectors $x,y\in\mathbb{R}^d$ are called \emph{neighboring}, sometimes also \emph{adjacent}, if
\[
\Vert x-y\Vert_1\le 1.
\]
\end{definition}
We remark that this definition is a generalization of the natural attribute-level privacy for binary input vectors, where privacy is preserved for a single bit-flip. For user-level privacy, we suppose that the contribution of a single user affects the $\ell_1$-norm of the input vector by at most 1. This is the case when we consider example histograms. More general user-level privacy is out of scope of this work.

\begin{definition}[Differential Privacy \cite{DworkMNS06,DworkKMMN06}]
A randomized mechanism $\mathcal{M}$ preserves $(\varepsilon,\delta)$-differential privacy, or \emph{approximate} differential privacy, if for any neighboring input vectors $x$ and $y$, and for all subsets $S\subset\operatorname{Range}(\mathcal{M})$, we have 
\[
\Pr[\mathcal{M}(x)\in S]\le e^\varepsilon\Pr[\mathcal{M}(y)\in S]+\delta.
\]
where the probability is over the random choices of $\mathcal{M}$. If $\delta=0$ we say that $\mathcal{M}$ preserves $\varepsilon$-differential privacy or \emph{pure} differential privacy.
\end{definition}
A common interpretation of approximate differential privacy is that we get pure differential privacy except with probability $\delta$ \cite{Mironov17}.

\subsubsection{Sensitivity}
\label{sec:sensitivitydef}
Dwork et al.~showed that we can obtain differential privacy by adding noise calibrated to the \emph{sensitivity} of a function \cite{DworkMNS06}. We define the sensitivity of a linear transformation:

\begin{definition}[$\ell_p$-sensitivity of transformation \cite{KenthapadiKMM13}]
For $p\ge 1$, the $\ell_p$-sensitivity of a linear transformation $S~:~\mathbb{R}^d\rightarrow \mathbb{R}^k$ is 
\begin{align*}
\Delta_p(S)&=\max_{\substack{x,y\in\mathbb{R}^d\\\Vert x-y\Vert_1\le 1}}\Vert Sx-Sy\Vert_p=\max_{1\le j\le d}\left(\sum_{i=1}^k\vert S_{i,j}\vert^{p}\right)^{1/p}=\max_{1\le j\le d}\Vert S_{\cdot,j}\Vert_p
\end{align*}
where $S_{\cdot,j}$ is the $j$\textsuperscript{th} column of $S$.
\end{definition}
\begin{note}
The definition follows from the observation that any vector of $\ell_1$-norm 1 (which is the case for neighboring vectors) can be represented as a convex linear combinations of basis vectors.
\end{note}

\subsubsection{Techniques In Differential Privacy}
We present two fundamental techniques in differential privacy that we use extensively in our analysis.
\label{sec:mechanisms}
\begin{lemma}[Laplace Mechanism \cite{DworkMNS06}]
\label{lem:LaplaceMech}
For linear transformation $S\in\mathbb{R}^{k\times d}$ and input $x\in\mathbb{R}^d$, the Laplace Mechanism with parameter $b$ outputs $Sx+\eta$ for $\eta\sim\operatorname{Lap}(0,b)^k$.
Let $\Delta_1$ be the $\ell_1$-sensitivity of $S$. The Laplace Mechanism with parameter $\Delta_1\varepsilon^{-1}$ preserves $\varepsilon$-differential privacy.
\end{lemma}

\begin{lemma}[Gaussian Mechanism \cite{DworkKMMN06,DworkR14}]
\label{lem:GaussMech}
For linear transform $S\in\mathbb{R}^{k\times d}$ and input $x\in\mathbb{R}^d$, the Gaussian Mechanism with parameter $\sigma$ outputs $Sx+\eta$ for $\eta\sim\mathcal{N}(0,\sigma^2)^k$.
Let $\Delta_2$ be the $\ell_2$-sensitivity of $S$. The Gaussian Mechanism with parameter $\sigma\ge \Delta_2\varepsilon^{-1}\sqrt{2\log(1.25/\delta)}$ preserves $(\varepsilon,\delta)$-differential privacy.
\end{lemma}

\subsection{Length Preserving Property}
\label{sec:lengthpreservingproperty}
Our technical results in Section \ref{sec:generalization} rely on linear transforms with the \emph{Length Preserving Property (LPP)}:
\begin{definition}[Length Preserving Property (LPP)]
\label{def:lengthpreservingproperty}
A random $k\times d$-projection $S$ satisfies the \emph{Length Preserving Property} if for any $x\in\mathbb{R}^d$ we have
\[
\operatorname{E}_{S}\left[\Vert Sx\Vert_2^2\right]=\Vert x\Vert_2^2.
\]
\end{definition}
Note that if $S$ satisfies LPP, then $S$ also preserves Euclidean distances and inner products, as $\langle x,y\rangle =\frac{\Vert x\Vert_2^2+\Vert y\Vert_2^2-\Vert x-y\Vert_2^2}{2}$.

\section{Technical Results}
\label{sec:generalization}
We now show our general, technical lemmas which will be useful for proving Theorem \ref{thm:main}. Let $S$ be a random $k\times d$-matrix with LPP as defined in Definition \ref{def:lengthpreservingproperty} and let $x,y\in\mathbb{R}^{d}$. Let $\mathcal{D}$ be a zero-mean distribution and $\eta,\mu\sim\mathcal{D}^k$ noise vectors. Let $\eta_*\sim \mathcal{D}$. 
We define
\[
\hat{E}_{gen}:=\Vert (Sx+\eta)-(Sy+\mu)\Vert_2^2-2k\operatorname{E}_{\mathcal{D}}[\eta_*^2].
\]
Our technical results are as follows:
\begin{restatable}{lemma}{technicalestimator}
\label{lem:technicalestimator}
We have
\begin{enumerate}
    \item $\hat{E}_{gen}$ is an unbiased estimator for $\Vert x-y\Vert_2^2$.
    \item The variance of $\hat{E}_{gen}$ is \begin{align*}
    \operatorname{Var}\left[\hat{E}_{gen}\right]&=\operatorname{Var}\left[\Vert Sx-Sy\Vert_2^2\right]+8\operatorname{E}_{\mathcal{D}}[\eta_*^2]\Vert x-y\Vert_2^2+2k\operatorname{E}_{\mathcal{D}}[\eta_*^4]+2k\operatorname{E}_{\mathcal{D}}[\eta_*^2]^2
    \end{align*}
\end{enumerate}
\end{restatable}
\begin{proof}
See Appendix \ref{app:technicallemmas}.
\end{proof}
Hence, the variance of $\hat{E}_{gen}$ is close to the variance of the non-private estimator, but has an additional noise term depending on the output dimension $k$ and the Euclidean distance of the input vectors. The following result describes the privacy guarantees of $\hat{E}_{gen}$:
\begin{restatable}{lemma}{technicalDP}
\label{lem:technicalDP}
Let $\Delta_1$ and $\Delta_2$ be the $\ell_1$- and $\ell_2$-sensitivities of $S$, respectively. 
Let $\delta>0$ be given and define 
\[
m:=\min\left\{\Delta_1,\Delta_2\sqrt{\ln\left(1/\delta\right)}\right\}.
\]
There is a distribution $\mathcal{D}$ such that
\begin{enumerate}
    \item The sketch $(S,Sx+\eta)$ is differentially private.
    \item $\operatorname{Var}\left[\hat{E}_{gen}\right]=\operatorname{Var}\left[\Vert Sx-Sy\Vert_2^2\right]+O\left(\frac{m^2}{\varepsilon^2}\Vert x-y\Vert_2^2+\frac{m^4}{\varepsilon^4}k\right).$
    \item Given $Sx$ and $Sy$, the estimate $\hat{E}_{gen}$ can be computed in time $O(k)$.
\end{enumerate}
\end{restatable}
\begin{proof}
We show that it suffices to let $\mathcal{D}$ be either the Normal or Laplace distribution for well-chosen parameters. We start with the following useful note:
\begin{note}
\label{note:expectations}
Let $n!!$ be the product of the numbers $1,...,n$ that have the same parity as $n$.
For $L\sim\operatorname{Lap}(b)$ and $G\sim\mathcal{N}(0,\sigma^2)$, we have 
\begin{align*}
&\forall n\in\mathbb{N}:\ \operatorname{E}_{\mathcal{D}}[L^n]=\frac{n!}{(b^{-1})^n}\\&\text{for even $n$:}\ \operatorname{E}_{\mathcal{D}}[G^n]=(n-1)!!\sigma^n.
\end{align*}
\end{note}

By Lemma \ref{lem:GaussMech}, the noisy projection $Sx+\eta$ is $(\varepsilon,\delta)$-differentially private for $\mathcal{D}=\mathcal{N}(0,\sigma^2)$ with $\sigma\ge \frac{\Delta_2}{\varepsilon}\sqrt{2\ln\left(1.25/\delta\right)}$. By the post-processing property of differential privacy, $\hat{E}_{gen}$ is also $(\varepsilon,\delta)$-differentially private.
From Lemma \ref{lem:technicalestimator} and Note \ref{note:expectations}
\begin{equation}
\begin{split}
\label{expr:noisetermGaussian}
    \operatorname{Var}\left[\hat{E}_{gen}\right]&=\operatorname{Var}\left[\Vert Sx-Sy\Vert_2^2\right]+O\left(\frac{\Delta_2^2\ln\left(\frac{1}{\delta}\right)}{\varepsilon^2}\Vert x-y\Vert_2^2+\frac{\Delta_2^4\ln^2\left(\frac{1}{\delta}\right)}{\varepsilon^4}k\right).
\end{split}
\end{equation}

Similarly, by Lemma \ref{lem:LaplaceMech} $\hat{E}_{gen}$ is $\varepsilon$-differentially private for $\mathcal{D}=\operatorname{Lap}(\Delta_1/\varepsilon)$, and from Lemma \ref{lem:technicalestimator} and Note \ref{note:expectations} we get
\begin{align}
\label{expr:noisetermLaplace}
    \operatorname{Var}\left[\hat{E}_{gen}\right]&=\operatorname{Var}\left[\Vert Sx-Sy\Vert_2^2\right]+O\left(\frac{\Delta_1^2}{\varepsilon^2}\Vert x-y\Vert_2^2+\frac{\Delta_1^4}{\varepsilon^4}k\right).
\end{align}
Finally, we can draw noise from the Laplace or the Normal distribution in constant time.
\end{proof}

\begin{note}
\label{note:technicalnoteDistributions}
As seen in the proof of Lemma \ref{lem:technicalDP}, letting $\mathcal{D}=Lap(\Delta_1/\varepsilon)$ gives $m=\Delta_1$ and letting $\mathcal{D}=\mathcal{N}(0,\sigma^2)$ for $\sigma\ge \Delta_2\varepsilon^{-1}\sqrt{2\ln(1.25/\delta)}$ gives $m=\Delta_2\sqrt{\ln(1/\delta)}$.
We wish to choose the $\mathcal{D}$ which minimizes $\operatorname{Var}[\hat{E}_{gen}]$. Ignoring constants, (\ref{expr:noisetermLaplace}) is upper bounded by (\ref{expr:noisetermGaussian}) when 
\begin{align}
\label{expr:Delta1vsDelta2}
\Delta_1<\Delta_2\sqrt{\ln(1/\delta)}\quad \Leftrightarrow\quad \delta<e^{-\Delta_1^2/\Delta_2^2}.
\end{align}
Hence, when (\ref{expr:Delta1vsDelta2}) is satisfied, we use $\mathcal{D}=\operatorname{Lap}(\Delta_1/\varepsilon)$ and otherwise let $\mathcal{D}=\mathcal{N}(0,\sigma^2)$ with $\sigma\ge \Delta_2\varepsilon^{-1}\sqrt{2\log(1.25/\delta)}$.
\end{note}

\section{Private Fast Johnson-Lindenstrauss Transform}
\label{sec:FJLT}
We now discuss a private version of the Fast Johnson-Lindenstrauss transform (FJLT) by Ailon \& Chazelle \cite{AilonC09}. We first remind the reader of the non-private transform in Section \ref{sec:FJLTDescription} and then give two private versions in Section \ref{sec:privateFJLT}. 

\subsection{Description of (non-private) Fast Johnson-Lindenstrauss Transform (FJLT)}
\label{sec:FJLTDescription}

We are concerned only with the transform preserving $\ell_2$-distances, but refer the reader to \cite{AilonC09} for the transform preserving $\ell_1$-distances as well as the analysis for the transforms.

FJLT is a random distribution of linear mappings $\Phi:\mathbb{R}^d\rightarrow \mathbb{R}^k$ with $k=O(\log(1/\beta)/\alpha^2)$, such that for $\alpha,\beta\in(0,1/2)$, with probability at least $1-\beta$
\[
(1-\alpha)k\Vert x\Vert_2^2\le \Vert \Phi x\Vert_2^2\le (1+\alpha)k\Vert x\Vert_2^2.
\]
For given values of $d, \alpha,\beta$, we describe how to obtain the random mapping $\Phi$ as the product of three real valued matrices, $P,H$ and $D$:  
\begin{itemize}
    \item $D$ is a random $d\times d$-diagonal matrix with $D_{jj}$ drawn independently from $\{-1,+1\}$ with probability $1/2$.
    \item $H$ is a $d\times d$-normalized Hadamard matrix such that for $f,j\in[d]$
    \[
    H_{fj}=\frac{1}{\sqrt{d}}(-1)^{\langle f-1,j-1\rangle}
    \]
    where $\langle f,j\rangle$ is the dot-product between vectors expressing $f$ and $j$ in binary representation.
    \item $P$ is a random $k\times d$-matrix whose entries are independently either normally distributed or 0. Specifically, for 
    \[
    q=\min\left\{\Theta\left(\frac{\log^2(1/\beta)}{d}\right),1\right\}
    \]
    we let $P_{if}$ be drawn (independently) from $\mathcal{N}(0,q^{-1})$ with probability $q$ and $P_{if}=0$ with probability $1-q$ for $i\in[k]$ and $f\in[d]$.
\end{itemize}
The transform $\Phi$ is defined as 
\[
\Phi:=PHD.
\]

To formalize, we get the following lemma:
\begin{lemma}[Lemma 2.1 from \cite{AilonC09}]
\label{lem:FJLT}
Let $\alpha,\beta\in(0,1/2)$ and let $\Phi$ be a random $k\times d$-projection matrix as described above. Let $x\in\mathbb{R}^d$. With probability at least $1-\beta$, the following two events occur:
\begin{itemize}
    \item $(1-\alpha)k\Vert x\Vert_2^2\le \Vert \Phi x\Vert_2^2\le (1+\alpha)k\Vert x\Vert_2^2$.
    \item The mapping $\Phi:\mathbb{R}^d\rightarrow \mathbb{R}^k$ requires time
    \[
    O\left(d\log d+dq\frac{\log(1/\beta)}{\alpha^{2}}\right)
    \]
    for $q=\min\left\{\Theta\left(\frac{\log^2(1/\beta)}{d}\right),1\right\}.$
\end{itemize}
\end{lemma}
\begin{proof}
See \cite{AilonC09}.
\end{proof}

We will henceforth concern ourselves with the \emph{normalized} FJLT, $1/\sqrt{k}\cdot\Phi$, such that 
\[
(1-\alpha)\Vert x\Vert_2^2\le 1/k\Vert \Phi x\Vert_2^2\le (1+\alpha)\Vert x\Vert_2^2.
\]

\begin{restatable}{lemma}{fjltLPP}
\label{lem:FJLTLPP}
The normalized FJLT satisfies LPP (see Definition \ref{def:lengthpreservingproperty}).
\end{restatable}
\begin{proof}
See Appendix \ref{app:FJLTLPP}.
\end{proof}

\begin{restatable}{lemma}{varianceFJLT}
\label{lem:varianceFJLT}
Let $x,y\in\mathbb{R}^d$ and let $\Phi$ be the FJLT as described above. Then 
\[
\Var[1/k\Vert \Phi x-\Phi y\Vert_2^2]\le \frac{3}{k}\Vert x-y\Vert_2^4.
\]
\end{restatable}
\begin{proof}
The proof follows directly from Lemma \ref{lem:varianceunderphi} in Appendix \ref{app:varianceunderphi}.
\end{proof}

\subsection{Private FJLT}
\label{sec:privateFJLT}
In this section, we argue how to construct a differentially private version of FJLT by adding Gaussian noise to the input.

If we simply exchange the i.i.d. normally distributed JL transform for FJLT in the work of Kenthapadi et al.~\cite{KenthapadiKMM13}, we get the following result.  Note that the $\ell_2$-sensitivity of the (normalized) projection is concentrated around 1, which justifies the choice of Gaussian noise.
\begin{corollary}
\label{cor:FJLToutput}
Let $\Phi$ be a random $k\times d$-FJLT and let $x,y\in\mathbb{R}^d$ be input vectors. Let $\Delta_2$ be the $\ell_2$-sensitivity of $\Phi$ and let $\eta,\mu\sim \mathcal{N}(0,\sigma^2)^k$ for $\sigma\ge \Delta_2 \varepsilon^{-1}\sqrt{2\log(1.25/\delta)}$ be noise vectors. Define
\[
\hat{E}_{FJLT_o}:=\frac{1}{k}\Vert (\Phi x+\eta)-(\Phi y+\mu)\Vert_2^2-2k\sigma^2
\]
\begin{itemize}
    \item $\hat{E}_{FJTL_o}$ is an unbiased estimator for $\Vert x-y\Vert_2^2$.
    \item The estimator has variance 
    \begin{align*}
    \Var\left[\hat{E}_{FJLT_o}\right]\le \frac{3}{k}\Vert x-y\Vert_2^4+O\left(k\sigma^4+\sigma^2\Vert x-y\Vert_2^2\right).
\end{align*}
    \item The sketch $(\Phi, \Phi x+\eta)$ is $(\varepsilon,\delta)$-differentially private.
    \item The sketch $\Phi x+\eta$ can be computed in time \[O\left(\max\left\{d\log d,\frac{dq\log(1/\beta)}{\alpha^2}\right\}\right)\] for $q=\min\{\Theta(\log^2(1/\beta)/d),1\}$.
    
\end{itemize}
\end{corollary}
\begin{proof}
That the estimator is unbiased and the variance follow from Lemmas \ref{lem:technicalestimator}, \ref{lem:FJLTLPP} and \ref{lem:varianceFJLT}. Privacy follows from Lemmas \ref{lem:GaussMech} and \ref{lem:technicalDP}. Running time follows from Lemmas \ref{lem:FJLT} and \ref{lem:technicalDP}.
\end{proof}

\begin{note}
\label{note:intializationFJLT}
Although the $\ell_2$-sensitivity of the normalized FJLT is concentrated around 1, the sensitivity of $\Phi$ \emph{could} (with a small probability) be very large, so the sketch $\Phi x+\eta$ suffers from the same initialization cost as the work of Kenthapadi et al.~(see Section \ref{sec:relatedworkKenthapadi}).
\end{note}

We now introduce a private version of FJLT, where we perturb the \emph{input}. This version avoids the issue described in Note \ref{note:intializationFJLT}, but will inevitably introduce error depending on the input size.

\begin{restatable}{lemma}{privateFJLTEstimator}
\label{lem:privateFJLTEstimator}
Let $\Phi$ be a random $k\times d$-FJLT and let $x,y\in\mathbb{R}^d$ be input vectors. Let $\eta,\mu\sim \mathcal{N}(0,\sigma^2)^d$ for $\sigma\ge \varepsilon^{-1}\sqrt{2\log(1.25/\delta)}$ be noise vectors. Define 
\[
\hat{E}_{FJLT_i}:=\frac{1}{k}\Vert \Phi(x+\eta)-\Phi(y+\mu)\Vert_2^2-2d\sigma^2
\]
\begin{itemize}
    \item $\hat{E}_{FJTL_i}$ is an unbiased estimator for $\Vert x-y\Vert_2^2$.
    \item The estimator has variance
    \begin{align*}
    \Var\left[\hat{E}_{FJLT_i}\right]\le \frac{3}{k}\Vert x-y\Vert_2^4+O\left(\frac{d^2\sigma^4}{k}+d\sigma^2\Vert x-y\Vert_2^2\right).
\end{align*}
    \item The sketch $(\Phi, \Phi(x+\eta))$ is $(\varepsilon,\delta)$-differentially private.
    \item The sketch $\Phi(x+\eta)$ can be computed in time \[O\left(\max\left\{d\log d,\frac{dq\log(1/\beta)}{\alpha^2}\right\}\right)\] for $q=\min\{\Theta(\log^2(1/\beta)/d),1\}$.
    
\end{itemize}
\end{restatable}
\begin{proof}
For proofs that the estimator is unbiased and for the variance, see Appendix \ref{app:omittedproofsPrivateFJLT}. We remark that the factor $d$ on the last term in the variance is a by-product of applying $\Phi$ to the noise.
Privacy follows directly from the Gaussian mechanism (see Lemma \ref{lem:GaussMech}), as the $\ell_2$-sensitivity is at most 1 (clearly, as we perturb the input vectors).
As noise can be added in time $O(d)$, the time required to compute the sketch follows from Lemma \ref{lem:FJLT}.
\end{proof}

\begin{note}
By spherical symmetry of the Normal distribution, $\Phi(x+\eta)$ and $\Phi x+P\eta$, where $P$ is defined in Section \ref{sec:FJLTDescription}, are identically distributed. Hence, one could add the same amount of noise after the Hadamard transform to get a differentially private sketch, that is, compute $P(HDx+\eta)$. Thus, for a given projection $P$, suppose column $j$ is all zeros, then we can immediately set $\eta_j=0$. This way, we may save a bit of randomness.
\end{note}

\section{Private Sparser Johnson-Lindenstrauss Transform}
\label{sec:applications}
In this section, we turn to the question of perturbation using Laplacian noise rather than Gaussian noise. We present and analyze a private sketch based on the SJLT, and conclude Theorem \ref{thm:main} in Section \ref{sec:summingup}. The main observation about this sketch is that we perturb the \emph{output} vectors rather than the input vectors while avoiding the initialization cost that was inherent to the work of Kenthapadi et al.~as well as Corollary \ref{cor:FJLToutput}. We compare the work of Kenthapadi et al., our private FJLT from Lemma \ref{lem:privateFJLTEstimator} and our private SJLT from Theorem \ref{thm:main} in Section \ref{sec:compareFJLTSparser}.

Theorem \ref{thm:main} is proven by combining the technical Lemmas \ref{lem:technicalestimator} and \ref{lem:technicalDP} with the SJLT.
Due to their sparsity, these transforms are more efficient than the suggestions from  \cite{KenthapadiKMM13}. We remark that this is just one example of linear transformations that our results can be applied to. It should also be noted that the results of Kenthapadi et al.~are directly transferable to the SJLT, although the results were only proven for the i.i.d. normally distributed JL transform, whereas we give the analysis here.

\subsection{Description of (non-private) Sparser Johnson-Lindenstrauss Transforms (SJLT)}
\label{sec:Kaneconstruction}
We first describe the SJLT from \cite{KaneN14}. We focus on the c)-construction and remark that similar arguments applies for the b)-construction. Let  $k=\Theta(\alpha^{-2}\log n)$ and let $x\in \mathbb{R}^d$ be an input vector. Let $h_1,...,h_s:[d]\rightarrow [k/s]$ and $\sign_1,...,\sign_s:[d]\rightarrow \{-1,+1\}$ be independent, random hash functions from $O(\log(1/\beta))$-wise independent families. Define $\xi_{ri}(j)=1[h_r(j)=i]$. Then $\E[\xi_{ri}(j)^2]=\E[\xi_{ri}(j)]=\frac{s}{k}$. The projection matrix $S$ is defined by
\[
S_{(i,r),j}=\frac{1}{\sqrt{s}}\sign_{r}(j)\xi_{ri}(j)
\]
for $i=1,...,k/s$ and $r=1,...,s$. Hence, entry $i'=i\cdot r\in[k]$ in the resulting embedding $Sx$ can be described as
\[
(Sx)_{i'}=(Sx)_{(i,r)}=\frac{1}{\sqrt{s}}\sum_{j=1}^d \sign_{r}(j)\xi_{ri}(j)x_j.
\]
We can think of $Sx$ as a vector consisting of $s$ blocks, each of length $k/s$. The $i$th block describes the projection of $x$ under $h_i$ and $\sign_i$.

\begin{restatable}{lemma}{sparserLPP}
\label{lem:sparserexpectationproperty}
The SJLT as described above satisfy LPP from Definition \ref{def:lengthpreservingproperty}.
\end{restatable}
\begin{proof}
The proof is a simple calculation and can be found in Appendix \ref{app:omittedproofsSparserLPP}.
\end{proof}

\begin{restatable}{lemma}{sparserJLvariance}
\label{lem:sparserJLnoiselessvariance}
Let $x,y\in\mathbb{R}^{d}$ and let $S$ be the SJLT as described above. Then 
\[
\operatorname{Var}\left[\Vert Sx-Sy\Vert_2^2\right]\le \frac{2}{k}\Vert x-y\Vert_2^4.
\]
\end{restatable}
\begin{proof}
The proof can be found in Appendix \ref{app:sparserJLvariance}.
\end{proof}

\subsection{Private SJLT}
We now turn to proving our main theorem, Theorem \ref{thm:main}. 
Combining Lemmas
\ref{lem:technicalestimator}, \ref{lem:sparserexpectationproperty} and \ref{lem:sparserJLnoiselessvariance}, we obtain the following corollary.

\begin{corollary}
\label{cor:sparserJL}
Let $S$ be a random $k\times d$-SJLT and let $x,y\in\mathbb{R}^{d}$ be input vectors. Let $\eta,\mu\sim \mathcal{D}^{k}$ be noise vectors where each entry is drawn from a zero-mean distribution $\mathcal{D}$. Then
\[
\hat{E}_{SJLT_\mathcal{D}}:=\Vert (Sx+\eta)-(Sy+\mu)\Vert_2^2-2k\operatorname{E}_{\mathcal{D}}[\eta_*^2]
\] is an unbiased estimator for $\Vert x-y\Vert_2^2$ with variance
\begin{equation}
\begin{split}
\label{expr:varianceKaneD}
\operatorname{Var}\left[\hat{E}_{SJLT_\mathcal{D}}\right]\le \frac{2}{k}\Vert x-y\Vert_2^4&+8\operatorname{E}_{\mathcal{D}}\left[\eta_*^2\right]\Vert x-y\Vert_2^2+2k\left(\operatorname{E}_{\mathcal{D}}[\eta_*^4]+\operatorname{E}_{\mathcal{D}}\left[\eta_*^2\right]^2\right).
\end{split}
\end{equation}
\end{corollary}

We have yet to choose $\mathcal{D}$ to ensure differential privacy of this estimator as well as argue about the efficiency. 
We first discuss the value of $k$.

\subsubsection{Optimal Projection Dimension}
\label{sec:optimalknoisy}
For the non-private SJLT, the optimal projection dimension is $k=\Theta\left(\alpha^{-2}\log(1/\beta)\right)$. One may ask what $k$ is optimal in the private case. The analysis and our optimal $k$ are very similar to the findings in \cite{KenthapadiKMM13}: we see that the 
variance in (\ref{expr:varianceKaneD}) is minimized for $k=\Theta\left(\frac{\Vert x-y\Vert_2^2}{\sqrt{\operatorname{E}[\eta_*^4]+\operatorname{E}[\eta_*^2]^2}}\right)$. By the same argument as in \cite{KenthapadiKMM13}, generally, no fixed value of $k$ will be optimal for the entire input domain, although there might be exceptions, when certain properties of the data are known. As in the work of Kenthapadi et al., if we have input domain $X$, then we may let $\nu=\max_{x\in X}\{\Vert x\Vert_2^2\}$ to obtain $k=\Theta(\nu\varepsilon^2/\Delta_1^2)$ for $\mathcal{D}=\operatorname{Lap}(\Delta_1/\varepsilon)$. Note that $k$ might not be optimal for all input vectors. We assume that $\nu$ is unknown and may be very large -- in particular, we consider vectors over the reals -- and thus proceed with $k=\Theta(\alpha^{-2}\log(1/\beta))$.

\subsubsection{Efficiency}
\label{sec:efficiency}
Let $S$ be a SJLT with sparsity $s=O\left(\alpha^{-1}\log(1/\beta)\right)$ and let input $x\in\mathbb{R}^d$ be given. The embedding $Sx$ can be computed in time $O(s\Vert x\Vert_0)$. Assuming that we sample from $\mathcal{N}(0,\sigma^2)$ and $\operatorname{Lap}(b)$ in constant time, random noise vector $\eta\sim\mathcal{D}$ for  $\mathcal{D}=\operatorname{Lap}(\Delta_1/\varepsilon)$ or $\mathcal{D}=\mathcal{N}(0,\sigma
^2)$ for $\sigma\ge \Delta_2 \varepsilon^{-1}\sqrt{2\log(1.25/\delta)}$ can be added in time $O(k)$ to give $Sx+\eta$. From \cite{securenoisegeneration} we know that we can at least sample from discretizations in expected constant time, so this assumption seems reasonable.
For given $Sx+\eta$ and $Sy+\mu$, the estimator $\hat{E}_{SJLT}$ can be computed in time $O(k)$.

\subsubsection{Summing Up}
\label{sec:summingup}
The SJLT as described in Section \ref{sec:Kaneconstruction}, where $k=\Theta(\alpha^{-2}\log(1/\beta))$ and $s=O(\alpha^{-1}\log(1/\beta))$, has $\ell_1$-sensitivity $\Delta_1=\sqrt{s}$ and $\ell_2$-sensitivity $\Delta_2=1$. Hence, consider Corollary \ref{cor:sparserJL} with $\mathcal{D}=\operatorname{Lap}(\sqrt{s}/\varepsilon)$. Lemma \ref{lem:LaplaceMech} ensures that $\hat{E}_{SJLT}$ is $\varepsilon$-differentially private. Combining with Section \ref{sec:efficiency} finishes the proof of Theorem \ref{thm:main}. If instead we let $\mathcal{D}=\mathcal{N}(0,\sigma^2)$ for $\sigma~\ge~\varepsilon^{-1}\sqrt{2\ln(1.25/\delta)}$ in Corollary \ref{cor:sparserJL}, $\hat{E}_{SJLT}$ is $(\varepsilon,\delta)$-differentially private and achieves the same variance as the work of Kenthapadi et al., while we gain a speed-up as well as avoid the initialization cost. Finally, we remark that by Note \ref{note:technicalnoteDistributions}, we minimize the variance of $\hat{E}_{SJLT}$ by letting $\mathcal{D}=\operatorname{Lap}(\sqrt{s}/\varepsilon)$ whenever $\delta<e^{-s}$.

\section{Comparison}
\label{sec:compareFJLTSparser}
We now compare Lemma \ref{lem:privateFJLTEstimator} and Theorem \ref{thm:main} with the work of Kenthapadi et al.

We first compare the running times to see for what parameters the private FJLT is faster than the private SJLT and then compare the variances for the two methods to get the speed-variance trade-off. Finally, we compare to the results of Kenthapadi et al.

Recall that our private FJLT can be computed in time
\[
O\left(\max\left\{d\log d,\frac{\log^3(1/\beta)}{\alpha^2}\right\}\right),
\]
and the private SJLT can be computed in time bounded by $O(sd)$ (for dense vectors) where $s=O\left(\log(1/\beta)\alpha^{-1}\right)$.
Observing that 
\[
O(sd)>O(d\log d)\quad \Leftrightarrow\quad d<e^{O(s)}=\frac{1}{\beta^{O(1/\alpha)}}
\]
and
\[
O(sd)>O\left(\frac{\log^3(1/\beta)}{\alpha^2}\right)\quad \Leftrightarrow\quad d>O\left(\frac{\log^2(1/\beta)}{\alpha}\right),
\]
we conclude that our private FJLT is indeed faster than the private SJLT whenever

\begin{align}
\label{interval:ford}
O\left(\frac{\log^2(1/\beta)}{\alpha}\right)<d<\frac{1}{\beta^{O(1/\alpha)}}.
\end{align}

We now turn to comparing the variances of the private versions of FJLT and SJLT:
Recall from Lemma \ref{lem:privateFJLTEstimator} that the private FJLT has variance 
\[
\Var\left[\hat{E}_{FJLT_i}\right]\le \frac{3}{k}\Vert x-y\Vert_2^4+O\left(d\sigma^2\Vert x-y\Vert_2^2+\frac{d^2\sigma^4}{k}\right).
\]
while, as seen in Theorem \ref{thm:main}, the private SJLT has variance
\[
\operatorname{Var}\left[\hat{E}_{SJLT}\right]\le \frac{2}{k}\Vert x-y\Vert_2^4+O\left(\frac{s}{\varepsilon^2}\Vert x-y\Vert_2^2+\frac{s^2}{\varepsilon^4}k\right).
\]
For the sake of simplicity, we will disregard the variance incurred by the transforms and limit ourselves to considering the terms incurred by the noise addition. The private SJLT (in particular) achieves a better variance than the private FJLT whenever
\begin{align*}
&O\left(\frac{d^2\sigma^4}{k}\right)=O\left(\frac{d^2\log^2(1/\delta)}{\varepsilon^4 k}\right)>O\left(\frac{s^2k}{\varepsilon^4}\right)\qquad \text{and}\\ &O\left(d\sigma^2\Vert x-y\Vert_2^2\right)=O\left(\frac{d\log(1/\delta)}{\varepsilon^2}\Vert x-y\Vert_2^2\right)>O\left(\frac{s}{\varepsilon^2}\Vert x-y\Vert_2^2\right).
\end{align*}
Treating each of the inequalities separately, we analyze for what values of $\delta$ this is the case:
\begin{align*}
    &O\left(\frac{d^2\log^2(1/\delta)}{\varepsilon^4 k}\right)>O\left(\frac{s^2k}{\varepsilon^4}\right)\quad\Leftrightarrow\quad \log(1/\delta)>O\left(\frac{sk}{d}\right)\quad \Leftrightarrow\quad \frac{1}{e^{O(sk/d)}}>\delta
\end{align*}
and
\begin{align*}
&O\left(\frac{d\log(1/\delta)}{\varepsilon^2}\Vert x-y\Vert_2^2\right)>O\left(\frac{s}{\varepsilon^2}\Vert x-y\Vert_2^2\right)\quad\Leftrightarrow\quad \log(1/\delta)>O\left(\frac{s}{d}\right)\quad\Leftrightarrow\quad \frac{1}{e^{O(s/d)}}>\delta.
\end{align*}
Hence, in particular, the private SJLT has smaller variance than the private FJLT whenever
\begin{align*}
\delta&<\min\left\{1/e^{O(s/d)}, 1/e^{O(sk/d)}\right\}=1/e^{O(sk/d)}=1/e^{O\left(\frac{\log^2(1/\beta)}{\alpha^{3}d}\right)}=\beta^{O\left(\frac{\log(1/\beta)}{\alpha^{3}d}\right)}.
\end{align*}

The variance of the estimator from Theorem \ref{thm:kenthapadiUnbiasedEstVar} by Kenthapadi et al.~was 
\[
\Var[\hat{E}_{iid}]=\frac{2}{k}\Vert x-y\Vert_2^4+O\left(\sigma^2\Vert x-y\Vert_2^2+\sigma^4k\right).
\]
An argument similar to the one above proves that the variance of our private SJLT improves over the variance of Kenthapadi et al.~when $\delta<e^{-s}=\beta^{O(1/\alpha)}$. Clearly,  Kenthapadi et al.~always achieves better variance than our private FJLT, due to the dependence on $d$ which was inherent from perturbing the input rather than the output, and we may assume $k<d$. 

Hence, we see a trade-off in running time versus variance, for certain values of input dimension $d$.

To sum up the above discussion, suppose that $\delta<\beta^{O(1/\alpha)}$. Then the private SJLT obtains the best variance out of all the methods. If $d$ satisfies (\ref{interval:ford}), then the private FJLT achieves the best running time, and otherwise, the private SJLT improves over the private FJLT in terms of both variance and running time.

\vspace{1cm}

\textbf{Acknowledgements\\}
This work was supported by Investigator Grant 16582, Basic Algorithms Research Copenhagen (BARC), from the VILLUM Foundation.
I would like to thank my advisor Rasmus Pagh for the support, great discussions and pointing out several interesting limitations of previous work, leading to this work.
I would also like to thank the reviewers for their excellent comments helping to improve this paper.

\bibliographystyle{plain}
\bibliography{bibl}

\appendix
\section{Omitted Proofs for Technical Lemmas}
\label{app:omittedproofs}

\label{app:technicallemmas}
\technicalestimator*
\begin{proof}
We start by showing 1). For simpler notation, we define $z:=x-y$. By independence and since $\operatorname{E}_{\mathcal{D}}[\eta_i]=0$ for all $i$, 
\begin{align*}
    \operatorname{E}_{S,\mathcal{D}}\Big[\big\Vert (Sx+\eta)-(Sy+\mu)\big\Vert_2^2\Big]&= \operatorname{E}_{S,\mathcal{D}}\left[\sum_{i=1}^k\left((Sx+\eta)_i-(Sy+\mu)_i\right)^2\right]\\&=\operatorname{E}_{S,\mathcal{D}}\left[\sum_{i=1}^k\left((Sz)_i^2+(\eta_i-\mu_i)^2+2(\eta_i-\mu_i)(Sz)_i\right)\right]\\
    &=\operatorname{E}_{S}\left[\sum_{i=1}^k(Sz)_i^2\right]+2\sum_{i=1}^k\operatorname{E}_{\mathcal{D}}\left[\eta_*^2\right]=\Vert z\Vert_2^2+2k\operatorname{E}_{\mathcal{D}}\left[\eta_*^2\right]
\end{align*}
where we in the last step used that $S$ has the LPP. 
So clearly, re-inserting $z=x-y$
\[
\operatorname{E}_{S,\mathcal{D}}\Big[\big\Vert (Sx+\eta)-(Sy+\mu)\big\Vert_2^2-2k\operatorname{E}_{\mathcal{D}}\left[\eta_*^2\right]\Big]=\Vert x-y\Vert_2^2.
\]

We turn to proving 2):
\begin{equation}
\begin{split}
\label{expr:varianceestimator}
    \operatorname{Var}\left[\hat{E}_{gen}\right]&=\operatorname{E}_{S,\mathcal{D}}\left[\hat{E}_{gen}^2\right]-\operatorname{E}_{S,\mathcal{D}}\left[\hat{E}_{gen}\right]^2=\operatorname{E}_{S,\mathcal{D}}\left[\hat{E}_{gen}^2\right]-\Vert x-y\Vert_2^4,
\end{split}
\end{equation}
so we analyze the first term:
\begin{equation}
\label{expr:generalvarianceterms}
\begin{split}
    \operatorname{E}_{S,\mathcal{D}}\left[\hat{E}_{gen}^2\right]&=\operatorname{E}_{S,\mathcal{D}}\left[\left(\big\Vert (Sx+\eta)-(Sy+\mu)\big\Vert_2^2-2k\operatorname{E}_{\mathcal{D}}\left[\eta_*^2\right]\right)^2\right]\\
    &=\operatorname{E}_{S,\mathcal{D}}\left[\big\Vert (Sx+\eta)-(Sy+\mu)\big\Vert_2^4\right]+4k^2\operatorname{E}_{\mathcal{D}}\left[\eta_*^2\right]^2\\&\qquad-4k\operatorname{E}_{\mathcal{D}}\left[\eta_*^2\right]\operatorname{E}_{S,\mathcal{D}}\left[\big\Vert (Sx+\eta)-(Sy+\mu)\big\Vert_2^2\right]
    \end{split}
\end{equation}
The last term in (\ref{expr:generalvarianceterms}) equals
\begin{equation}
\label{expr:lasttermgeneralvariancecomp}
\begin{split}
    &4k\operatorname{E}_{\mathcal{D}}\left[\eta_*^2\right]\left(\Vert x-y\Vert_2^2+2k\operatorname{E}_{\mathcal{D}}\left[\eta_*^2\right]\right)=4k\operatorname{E}_{\mathcal{D}}\left[\eta_*^2\right]\Vert x-y\Vert_2^2+8k^2\operatorname{E}_{\mathcal{D}}\left[\eta_*^2\right]^2
    \end{split}
\end{equation}

\begin{restatable}{claim}{claimfirstterm}
\label{claim:firstterm}
The first term in (\ref{expr:generalvarianceterms}) equals
\begin{align*}
    \operatorname{E}_{S,\mathcal{D}}\left[\big\Vert (Sx+\eta)-(Sy+\mu)\big\Vert_2^4\right]&=\operatorname{E}_{S}\left[\Vert S(x-y)\Vert_2^4\right]+4(k+2)\operatorname{E}_{\mathcal{D}}\left[\eta_*^2\right]\Vert x-y\Vert_2^2+2k\operatorname{E}_{\mathcal{D}}[\eta_*^4]\\&\qquad+2k(1+2k)\operatorname{E}_{\mathcal{D}}\left[\eta_*^2\right]^2
\end{align*}
\end{restatable}

The proof of the claim straightforward but tedious and thus left out here. It is proven formally in Appendix \ref{app:omittedproofsClaim}.

Inserting (\ref{expr:lasttermgeneralvariancecomp}) and Claim \ref{claim:firstterm} into (\ref{expr:generalvarianceterms}), we get
\begin{align*}
    \operatorname{E}_{S,\mathcal{D}}\left[\hat{E}_{gen}^2\right]&=\operatorname{E}_{S,\mathcal{D}}\left[\left(\big\Vert (Sx+\eta)-(Sy+\mu)\big\Vert_2^2-2k\operatorname{E}_{\mathcal{D}}\left[\eta_*^2\right]\right)^2\right]\\
    &=\operatorname{E}_{S}\left[\Vert S(x-y)\Vert_2^4\right]+4(k+2)\operatorname{E}_{\mathcal{D}}\left[\eta_*^2\right]\Vert x-y\Vert_2^2+2k\operatorname{E}_{\mathcal{D}}[\eta_*^4]+2k(1+2k)\operatorname{E}_{\mathcal{D}}\left[\eta_*^2\right]^2\\&\qquad+4k^2\operatorname{E}_{\mathcal{D}}\left[\eta_*^2\right]^2-4k\operatorname{E}_{\mathcal{D}}\left[\eta_*^2\right]\Vert x-y\Vert_2^2-8k^2\operatorname{E}_{\mathcal{D}}\left[\eta_*^2\right]^2\\
    &=\operatorname{E}_{S}\left[\Vert S(x-y)\Vert_2^4\right]+8\operatorname{E}_{\mathcal{D}}\left[\eta_*^2\right]\Vert x-y\Vert_2^2+2k\operatorname{E}_{\mathcal{D}}[\eta_*^4]+2k\operatorname{E}_{\mathcal{D}}\left[\eta_*^2\right]^2.
\end{align*}
Inserting this expression into (\ref{expr:varianceestimator}) proves that the variance is 
\begin{align*}
    \operatorname{Var}\left[\hat{E}_{gen}\right]&=\operatorname{E}_{S}\left[\Vert S(x-y)\Vert_2^4\right]+8\operatorname{E}_{\mathcal{D}}\left[\eta_*^2\right]\Vert x-y\Vert_2^2+2k\operatorname{E}_{\mathcal{D}}[\eta_*^4]+2k\operatorname{E}_{\mathcal{D}}\left[\eta_*^2\right]^2-\Vert x-y\Vert_2^4\\
    &=\operatorname{Var}\left[\Vert S(x-y)\Vert_2^2\right]+8\operatorname{E}_{\mathcal{D}}\left[\eta_*^2\right]\Vert x-y\Vert_2^2+2k\operatorname{E}_{\mathcal{D}}[\eta_*^4]+2k\operatorname{E}_{\mathcal{D}}\left[\eta_*^2\right]^2,
\end{align*}
again using that $S$ satisfies LPP.
\end{proof}

\subsubsection{Proof of Claim \ref{claim:firstterm}}
\label{app:omittedproofsClaim}
We repeat the claim for convenience:
\claimfirstterm*
\begin{proof}

For a simpler notation, we define $z:=x-y$. By simply unfolding the expression, we see that
\begin{align*}
    &=\sum_{i,\ell=1}^k\operatorname{E}_{S}\left[(Sz)_i^2(Sz)_\ell^2\right]+\sum_{i,\ell=1}^k\left(\operatorname{E}_{\mathcal{D}}\left[\eta_i^2\right]+\operatorname{E}_{\mathcal{D}}\left[\mu_i^2\right]\right)\operatorname{E}_{S}\left[(Sz)_\ell^2\right]\\&+0+\sum_{i,\ell=1}^k\operatorname{E}_{S}\left[(Sz)_i^2\right]\left(\operatorname{E}_{\mathcal{D}}\left[\eta_\ell^2\right]+\operatorname{E}_{\mathcal{D}}\left[\mu_\ell^2\right]\right)+\sum_{i=1}^k\operatorname{E}_{\mathcal{D}}\left[(\eta_i-\mu_i)^4\right]+\sum_{i\neq\ell}\operatorname{E}_{\mathcal{D}}\left[(\eta_i-\mu_i)^2\right]\operatorname{E}_{\mathcal{D}}\left[(\eta_\ell-\mu_\ell)^2\right]\\&+2\sum_{i=1}^k\operatorname{E}_{S}\left[(Sz)_i\right]\operatorname{E}_{\mathcal{D}}\left[(\eta_i-\mu_i)^3\right]+0+2\sum_{\ell=1}^k\operatorname{E}_{\mathcal{D}}\left[(\eta_\ell-\mu_\ell)^3\right]\operatorname{E}_{S}\left[(Sz)_\ell\right]+4\sum_{i=1}^k\operatorname{E}_{\mathcal{D}}\left[(\eta_i-\mu_i)^2\right]\operatorname{E}_{S}\left[(Sz)_i^2\right]\\&+4\sum_{i\neq\ell}\underbrace{\operatorname{E}_{\mathcal{D}}\left[(\eta_i-\mu_i)(\eta_\ell-\mu_\ell)\right]}_{=0}\operatorname{E}_{S}\left[(Sz)_i(Sz)_\ell\right]
\end{align*}
where we used that $\operatorname{E}[\eta_i]=\operatorname{E}[\mu_i]=0$ for all $i=1,...,k$ and that the noise is drawn independently of $S$.

Recalling that $\operatorname{E}_{\mathcal{D}}[\eta_i^2]~=~\operatorname{E}_{\mathcal{D}}[\mu_i^2]~=~\operatorname{E}_{\mathcal{D}}[\eta_*^2]$ for all $i$, we obtain
\begin{align*}
    &=\sum_{i,\ell=1}^k\operatorname{E}_{S}\left[(Sz)_i^2(Sz)_\ell^2\right]+4k\operatorname{E}_{\mathcal{D}}\left[\eta_*^2\right]\Vert z\Vert_2^2+k\left(2\operatorname{E}_{\mathcal{D}}[\eta_*^4]+6\operatorname{E}_{\mathcal{D}}[\eta_*^2]^2\right)+\sum_{i\neq\ell}4\operatorname{E}_{\mathcal{D}}\left[\eta_*^2\right]^2\\&+2\sum_{i=1}^k\operatorname{E}_{S}\left[(Sz)_i\right]\left(\operatorname{E}_{\mathcal{D}}\left[\eta_*^3\right]-\operatorname{E}_{\mathcal{D}}\left[\eta_*^3\right]\right)+2\sum_{\ell=1}^k\operatorname{E}_{S}\left[(Sz)_\ell\right]\left(\operatorname{E}_{\mathcal{D}}\left[\eta_*^3\right]-\operatorname{E}_{\mathcal{D}}\left[\eta_*^3\right]\right)\\&+4\sum_{i=1}^k2\operatorname{E}_{\mathcal{D}}\left[\eta_*^2\right]\operatorname{E}_{S}\left[(Sz)_i^2\right]
\end{align*}
which simplifies to
\begin{align*}
    &=\operatorname{E}_{S}\left[\Vert Sz\Vert_2^4\right]+4k\operatorname{E}_{\mathcal{D}}\left[\eta_*^2\right]\Vert z\Vert_2^2+2k\operatorname{E}_{\mathcal{D}}[\eta_*^4]+6k\operatorname{E}_{\mathcal{D}}[\eta_*^2]^2+4(k^2-k)\operatorname{E}_{\mathcal{D}}\left[\eta_*^2\right]^2+8\operatorname{E}_{\mathcal{D}}\left[\eta_*^2\right]\Vert z\Vert_2^2\\
    &=\operatorname{E}_{S}\left[\Vert Sz\Vert_2^4\right]+4(k+2)\operatorname{E}_{\mathcal{D}}\left[\eta_*^2\right]\Vert z\Vert_2^2+2k\operatorname{E}_{\mathcal{D}}[\eta_*^4]+2k(1+2k)\operatorname{E}_{\mathcal{D}}\left[\eta_*^2\right]^2
\end{align*}
Re-inserting $z=x-y$, we conclude that 
\begin{align*}
\operatorname{E}_{S,\mathcal{D}}\left[\big\Vert (Sx+\eta)-(Sy+\mu)\big\Vert_2^4\right]&=\operatorname{E}_{S}\left[\Vert S(x-y)\Vert_2^4\right]+4(k+2)\operatorname{E}_{\mathcal{D}}\left[\eta_*^2\right]\Vert x-y\Vert_2^2+2k\operatorname{E}_{\mathcal{D}}[\eta_*^4]\\&\qquad+2k(1+2k)\operatorname{E}_{\mathcal{D}}\left[\eta_*^2\right]^2.
\end{align*}
\end{proof}

\section{Omitted Proofs for FJLT}
\label{app:omittedproofsFJLT}
\subsection{Primitives}
\label{app:primitives}
This section will give some primitives that will be useful in the next section. Non-trivial arguments can be found in Section \ref{app:primitiveHelpers}.
We let $\Phi=PHD$ be the FJLT transform as described in Section \ref{sec:FJLTDescription} and $x,y\in\mathbb{R}^{d}$ be any real vectors. Let $X\sim\mathcal{N}(0,q^{-1})$. Then for any $i,n\in[k]$ and $j,\ell\in[d]$
\[
\E_P[P_{ij}]=0,\qquad \E_P\left[P_{ij}^2\right]=q\cdot \E_X\left[X^2\right]=1
\]
\[
\E_P\left[P_{ij}^4\right]=q\cdot \E_X\left[X^4\right]=q\cdot 3q^{-2}=\frac{3}{q}.
\]
\[
\E_D[D_{jj}]=0,\qquad \E_D[D_{jj}^2]=D_{jj}^2=1.
\]
\[
\E[\Phi_{ij}]=\sum_{f=1}^d\E_P\left[P_{if}\right]H_{fj}\E_D\left[D_{jj}\right]=0
\]
\[
\E_{\Phi}[\Phi_{ij}\Phi_{n\ell}]=\begin{cases}
\E_{\Phi}[\Phi_{ij}]\E_{\Phi}[\Phi_{n\ell}]=0,\qquad\qquad i\neq n,\ j\neq \ell\\
\E_{P}\left[(PH)_{ij}\right]\E_{P}\left[(PH)_{nj}\right]
=0,\ \ i\neq n,\ j=\ell\\
\E_{\Phi}[\Phi_{ij}]\E_{\Phi}[\Phi_{i\ell}]=0,\qquad\qquad i=n,\ j\neq \ell\\
\E_{\Phi}[\Phi_{ij}^2]=1,\qquad\qquad\qquad\quad\ \ i=n,\ j=\ell
\end{cases}
\]

\begin{align}
\label{eq:squaredphiterms}
\E_{\Phi}[\Phi_{ij}^2\Phi_{n\ell}^2]&=\begin{cases}
1,\qquad\qquad\qquad i\neq n\\
\frac{3}{qd}+1-\frac{3}{d},\qquad i=n,\ j\neq \ell\\
\frac{3}{qd}+3-\frac{3}{d},\qquad i=n,\ j= \ell
\end{cases}
\end{align}
 
\begin{align*}
    \E_{\Phi}[\Phi_{ij}\Phi_{i\ell}\Phi_{nv}\Phi_{nw}]=
    \begin{cases}
    \E_{\Phi}[\Phi_{ij}^2\Phi_{nv}^2],\qquad j=\ell,\ v=w\\
    \E_{\Phi}[\Phi_{ij}^2\Phi_{i\ell}^2],\qquad i=n,\ (j=v,\ \ell=w)\lor (j=w,\ \ell=v)\\
    0,\qquad\qquad\qquad\ \text{otherwise}
    \end{cases}
\end{align*}

\begin{align*}
    \E_{\Phi}[(\Phi x)_{i}(\Phi y)_{n}]=\sum_{j,\ell=1}^dx_jy_\ell\E_{\Phi}[\Phi_{ij}\Phi_{n\ell}]=\begin{cases}
    0,\qquad\qquad\qquad i\neq n\\
    \sum_{j=1}^dx_jy_j,\qquad\ i=n
    \end{cases}
\end{align*}

\begin{align}
\label{eq:squaredphix}
\E_{\Phi}[(\Phi x)_{i}^2(\Phi y)_{n}^2]=\Vert x\Vert_2^2\Vert y\Vert_2^2,\qquad i\neq n
\end{align}

\begin{equation}
\begin{split}
    \E_{\Phi}[(\Phi x)_{i}^2(\Phi y)_{i}^2]
    \label{eq:squaredphixieqn}
    &=
    \frac{3}{d}\left(\frac{d}{3}+\left(\frac{1}{q}-1\right)\right)\left(\Vert x\Vert_2^2\Vert y\Vert_2^2+2\langle x,y\rangle^2\right)-\frac{6}{d}\left(\frac{1}{q}-1\right)\sum_{j=1}^dx_j^2y_j^2
\end{split}
\end{equation}

\subsubsection{Arguments for Primitives}
\label{app:primitiveHelpers}
\subsubsection*{Argument for (\ref{eq:squaredphiterms})}
We use that
\begin{align*}
\E_{\Phi}[\Phi_{ij}^2\Phi_{n\ell}^2]=\sum_{f,g,h,s=1}^d\E_{P}[P_{if}P_{ig}P_{nh}P_{ns}]H_{fj}H_{gj}H_{h\ell}H_{s\ell}\E_{D}[D_{jj}^2D_{\ell\ell}^2]
\end{align*}
and so for $i\neq n$
\[
\sum_{f,h=1}^d\E_{P}\left[P_{if}^2\right]\E_{P}\left[P_{nh}^2\right]H_{fj}^2H_{h\ell}^2=1
\]
and for $i=n$
\begin{align*}
    \sum_{f,g,h,s=1}^d\E_{P}[P_{if}P_{ig}P_{nh}P_{ns}]H_{fj}H_{gj}H_{h\ell}H_{s\ell}&=\sum_{f=1}^d\E_{P}[P_{if}^4]H_{fj}^2H_{f\ell}^2+\sum_{f\neq h=1}^d\E_{P}[P_{if}^2]\E_{P}[P_{ih}^2]H_{fj}^2H_{h\ell}^2\\&\qquad+2\sum_{f\neq \ell=1}^d\E_{P}[P_{if}^2]\E_{P}[P_{ig}^2]H_{fj}H_{gj}H_{f\ell}H_{g\ell}\\
    &=\E_{P}[P_{if}^4]/d+(d^2-d)/d^2+2(\langle H_j,H_\ell\rangle^2-1/d)\\
    &=\E_{P}[P_{if}^4]/d+1-3/d+2\langle H_j,H_\ell\rangle^2\\&=\begin{cases}
    \E_{P}[P_{if}^4]/d+1-3/d,\qquad j\neq \ell\\
\E_{P}[P_{if}^4]/d+3-3/d,\qquad j=\ell
    \end{cases}
\end{align*}
because $\langle H_j,H_\ell\rangle=\begin{cases}
0,\qquad j\neq \ell\\
1,\qquad j=\ell.
\end{cases}$

\subsubsection*{Argument for (\ref{eq:squaredphix}) and (\ref{eq:squaredphixieqn})}
We used that 
\begin{align*}
    &\E_{\Phi}[(\Phi x)_{i}^2(\Phi y)_{n}^2]=\sum_{j,v=1}^dx_j^2y_v^2\E_{\Phi}[\Phi_{ij}^2\Phi_{nv}^2]+2\sum_{j\neq \ell=1}^dx_j^2y_\ell^2\E_{\Phi}[\Phi_{ij}\Phi_{i\ell}\Phi_{nj}\Phi_{n\ell}]
\end{align*}
where we for $i\neq n$ get $\Vert x\Vert_2^2\Vert y\Vert_2^2$ and for $i=n$ get 
\begin{align*}
    &\sum_{j=1}^dx_j^2y_j^2\E_{\Phi}[\Phi_{ij}^4]+\sum_{j\neq v=1}^dx_j^2y_v^2\E_{\Phi}[\Phi_{ij}^2\Phi_{iv}^2]+2\sum_{j\neq \ell=1}^dx_jx_\ell y_jy_\ell\E_{\Phi}[\Phi_{ij}^2\Phi_{i\ell}^2]\\&=\sum_{j=1}^dx_j^2y_j^2\left(\frac{3}{qd}+3-\frac{3}{d}\right)+\sum_{j\neq v=1}^dx_j^2y_v^2\left(\frac{3}{qd}+1-\frac{3}{d}\right)+2\sum_{j\neq \ell=1}^dx_jx_\ell y_jy_\ell\left(\frac{3}{qd}+1-\frac{3}{d}\right)\\&=\sum_{j=1}^dx_j^2y_j^2\left(\frac{3}{qd}+3-\frac{3}{d}\right)+\left(\Vert x\Vert_2^2\Vert y\Vert_2^2-\sum_{j=1}^dx_j^2y_j^2\right)\left(\frac{3}{qd}+1-\frac{3}{d}\right)+2\left(\langle x,y\rangle^2-\sum_{j=1}^dx_j^2y_j^2\right)\left(\frac{3}{qd}+1-\frac{3}{d}\right)\\&=\frac{3}{d}\left(\frac{d}{3}+\left(\frac{1}{q}-1\right)\right)\left(\Vert x\Vert_2^2\Vert y\Vert_2^2+2\langle x,y\rangle^2\right)-\frac{6}{d}\left(\frac{1}{q}-1\right)\sum_{j=1}^dx_j^2y_j^2
\end{align*}

\subsection{Proof of FJLT Satisfying LPP}
\label{app:FJLTLPP}
\fjltLPP*
\begin{proof}
Applying the primitives from Appendix \ref{app:primitives}, we get 
\begin{align*}
    \E_{\Phi}\left[\frac{1}{k}\Vert \Phi x\Vert_2^2\right]&=\frac{1}{k}\E_{\Phi}\left[\sum_{i=1}^k(\Phi x)_i^2\right]=\frac{1}{k}\sum_{i=1}^k\sum_{j,\ell=1}^d\E_{\Phi}\left[\Phi_{ij}\Phi_{i\ell}\right]x_jx_\ell=\frac{1}{k}\sum_{i=1}^k\sum_{j=1}^d\E_{\Phi}\left[\Phi_{ij}^2\right]x_j^2=\Vert x\Vert_2^2.
\end{align*}
\end{proof}

\subsection{Variance under FJLT}
\label{app:varianceunderphi}
For convenience, we prove the following result, as it will be useful in this form for several other proofs. Note that Lemma \ref{lem:varianceFJLT} follows directly from Lemma \ref{lem:varianceunderphi}.
\begin{lemma}
\label{lem:varianceunderphi}
Let $k\times d$-matrix $\Phi=PHD$, where $P_{ij}$ is $\mathcal{N}(0,q^{-1})$ with probability $q$ and and 0 otherwise. For input vector $\eta\sim\mathcal{D}^d$ for a real-valued distribution $\mathcal{D}$:
\[
\Var[\Vert \Phi \eta\Vert_2^2]\le \frac{3}{k}\E_{\eta}\left[\Vert \eta\Vert_2^2\right].
\]
For $x\in \mathbb{R}^d$, we get 
\[
\Var[\Vert \Phi x\Vert_2^2]\le \frac{3}{k}\Vert x\Vert_2^2.
\]
\end{lemma}
\begin{proof}
\begin{align*}
    \Var[\Vert \Phi \eta\Vert_2^2]&= \E_{\Phi,\eta}[\Vert \Phi \eta\Vert_2^4]- \E_{\Phi,\eta}[\Vert \Phi \eta\Vert_2^2]^2=\sum_{i,n=1}^k\E_{\Phi,\eta}\left[(\Phi \eta)_i^2(\Phi \eta)_n^2\right]- k^2\E_{\eta}\left[\Vert \eta\Vert_2^4\right]\\&=
    \sum_{i=1}^k\E_{\Phi,\eta}\left[(\Phi  \eta)_i^4\right]+\sum_{i\neq n=1}^k\E_{\Phi}\left[(\Phi \eta)_i^2(\Phi \eta)_n^2\right]- k^2\E_{\eta}\left[\Vert \eta\Vert_2^4\right]\\
    &=\frac{9k}{d}\left(\frac{d}{3}+\left(\frac{1}{q}-1\right)\right)\E_{\eta}\left[\Vert \eta\Vert_2^4\right]-\frac{6k}{d}\left(\frac{1}{q}-1\right)\E_{\eta}\left[\Vert \eta\Vert_4^4\right]-k^2\E_{\eta}\left[\Vert \eta\Vert_2^4\right]\\
    &=3k\left(\frac{2}{3}+\frac{3}{d}\left(\frac{1}{q}-1\right)\right)\E_{\eta}\left[\Vert \eta\Vert_2^4\right]-\frac{6k}{d}\left(\frac{1}{q}-1\right)\E_{\eta}\left[\Vert \eta\Vert_4^4\right].
\end{align*}
which again implies
\begin{align*}
    \Var\left[\frac{1}{k}\Vert \Phi \eta\Vert_2^2\right]&=\frac{3}{k}\left(\frac{2}{3}+\frac{3}{d}\left(\frac{1}{q}-1\right)\right)\E_{\eta}\left[\Vert \eta\Vert_2^4\right]-\frac{6}{dk}\left(\frac{1}{q}-1\right)\E_{\eta}\left[\Vert \eta\Vert_4^4\right]\\&\le \frac{3\E_{\eta}\left[\Vert \eta\Vert_2^4\right]}{k}\left(
    \frac{2}{3}+\frac{3}{d}\left(\frac{1}{q}-1\right)\right)\le \frac{3}{k}\E_{\eta}\left[\Vert \eta\Vert_2^4\right]
\end{align*}
when $q\ge \frac{1}{d/9+1}$. 
\end{proof}

\section{Omitted Proofs for Private FJLT}
\subsection{Estimator and Variance for Private FJLT}
\label{app:omittedproofsPrivateFJLT}

\begin{lemma}
We have 
\begin{enumerate}
    \item $\hat{E}_{FJLT_i}$ is an unbiased estimator for $\Vert x-y\Vert_2^2$.
    \item $\Var[\hat{E}_{FJLT_i}]\le \frac{3}{k}\Vert x-y\Vert_2^4+O\left(\frac{d^2\sigma^4}{k}+d\sigma^2\Vert x-y\Vert_2^2\right)$.
\end{enumerate}
\end{lemma}
\begin{proof}
We repeatedly apply the primitives of Section \ref{app:primitives} and Section \ref{app:varianceunderphi}.

We start by proving 1). Observe that
\begin{align*}
    \E\left[\Vert \Phi(x+\eta)-\Phi(y+\mu)\Vert_2^2\right]&=\E\left[\Vert \Phi(x-y)+\Phi(\eta-\mu)\Vert_2^2\right]\\&=\E\left[\Vert \Phi(x-y)\Vert_2^2\right]+\E\left[\Vert\Phi(\eta-\mu)\Vert_2^2\right]\\&=k\Vert x-y\Vert_2^2+k\E_{\eta,\mu}\left[\Vert \eta-\mu\Vert_2^2\right]
\end{align*}
Since $\eta,\mu\sim\mathcal{N}(0,\sigma^2)^d$, we have $\eta-\mu\sim\mathcal{N}(0,2\sigma^2)^d$ and so
\[
\E_{\eta}\left[\Vert \eta-\mu\Vert_2^2\right]=\sum_{j=1}^d\E_{\eta,\mu}[(\eta_j-\mu_j)^2]=2d\sigma^2.
\]
We conclude that 
\[
\hat{E}_{FJLT_i}=1/k\Vert \Phi(x+\eta)-\Phi(y+\mu)\Vert_2^2-2d\sigma^2
\]
is an unbiased estimator for $\Vert x-y\Vert_2^2$.

We turn to proving 2). Note that
\begin{align*}
&\Var\left[1/k\Vert \Phi (x-y)+\Phi (\eta-\mu)\Vert_2^2-2d\sigma^2\right]=\frac{1}{k^2}\Var\left[\Vert \Phi (x-y)+\Phi (\eta-\mu)\Vert_2^2\right],
\end{align*}
so it suffices to consider the RHS. For readability, we will do the analysis for $x$ and $\eta$, and eventually substitute $x$ for $x-y$ and $\eta$ for $\eta-\mu$, recalling that if $\eta\sim \mathcal{N}(0,\sigma^2)$, then $\eta-\mu\sim\mathcal{N}(0,2\sigma^2)$.

For any $x,\eta\in\mathbb{R}^d$
\begin{align*}
    \E_{\Phi,\eta}\left[\Vert \Phi x+\Phi \eta\Vert_2^2\right]^2&=\left(\E_{\Phi}\left[\Vert \Phi x\Vert_2^2\right]+\E_{\Phi}\left[\Vert\Phi\eta\Vert_2^2\right]\right)^2\\
    &=\E_{\Phi}\left[\Vert \Phi x\Vert_2^2\right]^2+\E_{\Phi}\left[\Vert\Phi\eta\Vert_2^2\right]^2+2\E_{\Phi}\left[\Vert \Phi x\Vert_2^2\right]\E_{\Phi}\left[\Vert\Phi \eta\Vert_2^2\right]
\end{align*}

By the triangle inequality, we see that
\begin{align*}
    \E_{\Phi,\eta}\left[\Vert \Phi x+\Phi\eta\Vert_2^4\right]&=\E_{\Phi,\eta}\left[\left(\Vert \Phi x+\Phi\eta\Vert_2^2\right)^2\right]\\&\le \E_{\Phi,\eta}\left[\left(\Vert \Phi x\Vert_2^2+\Vert\Phi\eta\Vert_2^2+2\Vert \Phi x\Vert_2\Vert \Phi\eta\Vert_2\right)^2\right]\\
    &= \E_{\Phi}\left[\Vert \Phi x\Vert_2^4\right]+\E_{\Phi,\eta}\left[\Vert\Phi\eta\Vert_2^4\right]+6\E_{\Phi,\eta}\left[\Vert \Phi x\Vert_2^2\Vert \Phi\eta\Vert_2^2\right]
\end{align*}
Where the last equality follows from the zero-meaned $\eta$ leading to a several terms cancelling out.

Hence, the variance is bounded by
\begin{align*}
    \Var\left[\Vert \Phi x+\Phi \eta\Vert_2^2\right]&\le\E_{\Phi}\left[\Vert \Phi x\Vert_2^4\right]+\E_{\Phi,\eta}\left[\Vert\Phi\eta\Vert_2^4\right]+6\E_{\Phi,\eta}\left[\Vert \Phi x\Vert_2^2\Vert \Phi\eta\Vert_2^2\right]-\E_{\Phi}\left[\Vert \Phi x\Vert_2^2\right]^2-\E_{\Phi,\eta}\left[\Vert\Phi\eta\Vert_2^2\right]^2\\&\qquad-2\E_{\Phi}\left[\Vert \Phi x\Vert_2^2\right]\E_{\Phi,\eta}\left[\Vert\Phi\eta\Vert_2^2\right]
\end{align*}
which again implies
\begin{equation}
\begin{split}
    \Var\left[1/k\Vert \Phi x+\Phi \eta\Vert_2^2\right]&\le 
    \label{eq:varianceofunbiasedestimatorFJLT}\Var_{\Phi}\left[1/k\Vert \Phi x\Vert_2^2\right]+\Var_{\Phi,\eta}\left[1/k\Vert\Phi\eta\Vert_2^2\right]+\frac{6}{k^2}\E_{\Phi,\eta}\left[\Vert \Phi x\Vert_2^2\Vert \Phi\eta\Vert_2^2\right]\\&\qquad-\frac{2}{k^2}\E_{\Phi}\left[\Vert \Phi x\Vert_2^2\right]\E_{\Phi,\eta}\left[\Vert\Phi\eta\Vert_2^2\right]
\end{split}
\end{equation}

For the last term we have
\begin{align*}
    \E_{\Phi}\left[\Vert \Phi x\Vert_2^2\right]\E_{\Phi,\eta}\left[\Vert\Phi\eta\Vert_2^2\right]=2k^2\Vert x\Vert_2^2\E_{\eta}[\Vert \eta\Vert_2^2]=2k^2d\sigma^2\Vert x\Vert_2^2
\end{align*}

and for the second to last term we get:
\begin{align*}
    \E_{\Phi,\eta}\left[\Vert \Phi x\Vert_2^2\Vert \Phi\eta\Vert_2^2\right]&=\sum_{i,n=1}^k\E_{\Phi,\eta}[(\Phi x)_i^2(\Phi \eta)_n^2]=\sum_{i=1}^k\E_{\Phi,\eta}[(\Phi x)_i^2(\Phi \eta)_i^2]+\sum_{i\neq n=1}^k\E_{\Phi,\eta}[(\Phi x)_i^2(\Phi \eta)_n^2]\\&=\frac{3k}{d}\left(\frac{d}{3}-\left(1-\frac{1}{q}\right)\right)\left(\Vert x\Vert_2^2\E\left[\Vert \eta\Vert_2^2\right]+2\E\left[\langle x,\eta\rangle^2\right]\right)+\frac{6k}{d}\left(1-\frac{1}{q}\right)\sum_{j=1}^dx_j^2\E\left[\eta_j^2\right]\\&\qquad+(k^2-k)\Vert x\Vert_2^2\E_{\eta}[\Vert \eta\Vert_2^2]\\&=\frac{3k}{d}\left(\frac{d}{3}-\left(1-\frac{1}{q}\right)\right)\left(\Vert x\Vert_2^2\E\left[\Vert \eta\Vert_2^2\right]+2\Vert x\Vert_2^2\E_{\eta}[\eta_*^2]\right)+\frac{6k}{d}\left(1-\frac{1}{q}\right)\Vert x\Vert_2^2\E\left[\eta_*^2\right]\\&\qquad+(k^2-k)\Vert x\Vert_2^2\E_{\eta}[\Vert \eta\Vert_2^2]\\&=k\left(k-\frac{3}{d}\left(1-\frac{1}{q}\right)\right)\Vert x\Vert_2^2d\sigma^2+2k\Vert x\Vert_2^2\sigma^2\\&=k\Vert x\Vert_2^2\sigma^2\left(kd+2-3\left(1-\frac{1}{q}\right)\right)
\end{align*}

Inserting into (\ref{eq:varianceofunbiasedestimatorFJLT}) and applying Section \ref{app:varianceunderphi}, we see that 
\begin{align*}
    \Var[1/k\Vert \Phi x+\Phi \eta\Vert_2^2]&\le \frac{3}{k}\Vert x\Vert_2^4+O\left(\frac{d^2\sigma^4}{k}\right)+\frac{6}{k}\Vert x\Vert_2^2\sigma^2\left(kd+2-3\left(1-\frac{1}{q}\right)\right)-4d\sigma^2\Vert x\Vert_2^2\\&=\frac{3}{k}\Vert x\Vert_2^4+O\left(\frac{d^2\sigma^4}{k}\right)+\frac{2}{k}\Vert x\Vert_2^2\sigma^2\left(kd+6+9\left(\frac{1}{q}-1\right)\right)
\end{align*}
Substituting $x$ for $x-y$ and $\eta$ for $\eta-\mu$ proves that
\begin{align*}
    \Var[1/k\Vert \Phi(x+\eta)-\Phi(y+\mu)\Vert_2^2-2d\sigma^2]&\le \frac{3}{k}\Vert x-y\Vert_2^4+O\left(\frac{d^2\sigma^4}{k}+d\sigma^2\Vert x\Vert_2^2+\frac{\sigma^2}{qk}\Vert x\Vert_2^2\right)
\end{align*}
Recalling that $q= \min\left\{\Theta\left(\frac{\log k}{d}\right),1\right\}$ we get 
\[
\frac{3}{k}\Vert x-y\Vert_2^4+O\left(\frac{d^2\sigma^4}{k}+d\sigma^2\Vert x\Vert_2^2\right)
\]
concluding the proof.
\end{proof}

\section{Omitted Proofs for SJLT}
\subsection{Proof of SJLT Satisfying LPP}
\label{app:omittedproofsSparserLPP}
\sparserLPP*
\begin{proof}
We show the result here for the $c)$-construction. A similar proof shows the result for the $b)$-construction. 

\begin{align*}
\operatorname{E}_{S}\left[\Vert Sx\Vert_2^2\right]&=\operatorname{E}_{S}\left[\sum_{i=1}^{k/s}\sum_{r=1}^s(Sx)_{(i,r)}^2\right]=\frac{1}{s}\operatorname{E}_{h,\sign}\left[\sum_{i=1}^{k/s}\sum_{r=1}^s\left(\sum_{j=1}^d \sign_{r}(j)\xi_{ri}(j)x_j\right)^2\right]\\&=\frac{1}{s}\operatorname{E}_{S}\left[\sum_{i=1}^{k/s}\sum_{r=1}^s\sum_{j,\ell=1}^d \sign_{r}(j)\sign_{r}(\ell)\xi_{ri}(j)\xi_{ri}(\ell)x_jx_\ell\right]\\&=\frac{1}{s}\sum_{j=1}^dx_j^2\sum_{i=1}^{k/s}\sum_{r=1}^s \operatorname{E}_{h}\left[\xi_{ri}(j)\right]=\Vert x\Vert_2^2
\end{align*}
because $\sign_r(j)$ and $\sign_r(\ell)$ are independent for $j\neq \ell$ and $\operatorname{E}_{\sign}[\sign_r(j)]=0$.
\end{proof}

\subsection{Proof of Variance of (non-private) SJLT}
\label{app:sparserJLvariance}
The following lemma will be useful throughout this appendix. The proof is immediate from the definition of $\xi$.
\begin{lemma}
\label{lem:primitivesxi}
\begin{align*}
    \E_{\xi}[\xi_{ri}(j)\xi_{tn}(\ell)]&=\begin{cases}
    \E_{\xi}[\xi_{ri}(j)]\E_{\xi}[\xi_{tn}(\ell)],\quad j\neq \ell\\
    \E_{\xi}[\xi_{ri}(j)^2],\qquad\qquad\quad r=t,\ i=n,\ j=\ell\\
    \E_{\xi}[\xi_{ri}(j)]\E_{\xi}[\xi_{ti}(j)],\quad r\neq t,\ i=n,\ j=\ell\\
    0,\qquad\qquad\qquad\qquad\quad\ r= t,\ i\neq n,\ j=\ell\\
    \E_{\xi}[\xi_{ri}(j)]\E_{\xi}[\xi_{tn}(j)],\quad r\neq t,\ i\neq n,\ j=\ell\\
    \end{cases}\\&=\begin{cases}
    s^2/k^2,\qquad j\neq \ell\\
    s/k,\qquad\quad r=t,\ i=n,\ j=\ell\\
   s^2/k^2,\qquad r\neq t,\ i=n,\ j=\ell\\
    0,\qquad\qquad r= t,\ i\neq n,\ j=\ell\\
    s^2/k^2,\qquad r\neq t,\ i\neq n,\ j=\ell
    \end{cases}
\end{align*}
where we recalled that
\[
h_{r}(j)=i\land h_{r}(j)=n,\qquad \Leftrightarrow\qquad i=n.
\]
\end{lemma}

We now prove Lemma \ref{lem:sparserJLnoiselessvariance}. We state it here for convenience.
\sparserJLvariance*

\begin{proof}
Throughout the proof, we will apply Lemma \ref{lem:primitivesxi} without further comment.
By linearity of $S$ and since $S$ satisfies LPP, it is sufficient to show that for $x\in \mathbb{R}^d$
\begin{align*}
    \operatorname{Var}\left[\Vert Sx\Vert_2^2\right]&=\operatorname{E}_{S}\left[\Vert Sx\Vert_2^4\right]-\left(\operatorname{E}_{S}\left[\Vert Sx\Vert_2^2\right]\right)^2=\operatorname{E}_{S}\left[\Vert Sx\Vert_2^4\right]-\Vert x\Vert_2^4\le \frac{2}{k}\Vert x\Vert_2^4.
\end{align*}
We will consider the first term:
\begin{align}
    \nonumber
    \operatorname{E}_{S}\left[\Vert Sx\Vert_2^4\right]&\nonumber=\operatorname{E}_{S}\left[\left(\sum_{i=1}^{k/s}\sum_{r=1}^s(Sx)_{(i,r)}^2\right)^2\right]=\operatorname{E}_{S}\left[\left(\sum_{i=1}^{k/s}\sum_{r=1}^s\left(\sum_{j=1}^d\frac{1}{\sqrt{s}}x_j\sign_r(j)\xi_{ri}(j)\right)^2\right)^2\right]\\
    \label{expr:SparserSecondmoment}
    &=\frac{1}{s^2}\operatorname{E}_{S}\left[\left(\sum_{i=1}^{k/s}\sum_{r=1}^s\sum_{j,\ell=1}^dx_jx_\ell\sign_r(j)\sign_r(\ell)\xi_{ri}(j)\xi_{ri}(\ell)\right)^2\right]
\end{align}
Letting
\[
a=\sum_{i=1}^{k/s}\sum_{r=1}^s\sum_{j=1}^dx_j^2\xi_{ri}(j)
\]
and 
\[
b=\sum_{i=1}^{k/s}\sum_{r=1}^s\sum_{j\neq\ell}x_jx_\ell\sign_r(j)\sign_r(\ell)\xi_{ri}(j)\xi_{ri}(\ell).
\]
we can express (\ref{expr:SparserSecondmoment}) as 
\begin{align}
    \label{expr:expectationAB}
    \frac{1}{s^2}\operatorname{E}_{S}\left[\left(a+b\right)^2\right]&=\frac{1}{s^2}\operatorname{E}_{S}\left[a^2+b^2+2ab\right]
\end{align}
The proofs of the following claims are straightforward but tedious and thus we leave them out here. They can be found in Appendix \ref{app:omittedproofsSparserclaims}.
\begin{restatable}{claim}{asquared}
\label{claim:asquared}
\[
\operatorname{E}_{S}\left[a^2\right]=s^2\Vert x\Vert_2^4.
\]
\end{restatable}
\begin{restatable}{claim}{bsquared}
\label{claim:bsquared}
\[
\operatorname{E}_{S}\left[b^2\right]=\frac{2s^2}{k}\left(\Vert x\Vert_2^4-\Vert x\Vert_4^4\right).
\]
\end{restatable}
\begin{restatable}{claim}{twoab}
\label{claim:twoab}
\[
2\operatorname{E}_{S}\left[ab\right]=0
\]
\end{restatable}
Inserting Claims \ref{claim:asquared}-\ref{claim:twoab} into (\ref{expr:expectationAB}), we conclude that
\begin{align*}
    \operatorname{E}_{S}\left[\Vert Sx\Vert_2^4\right]=\Vert x\Vert_2^4+\frac{2}{k}\left(\Vert x\Vert_2^4-\Vert x\Vert_4^4\right)
\end{align*}
finally proving that
\begin{align*}
    \operatorname{Var}\left[\Vert Sx\Vert_2^2\right]=\frac{2}{k}\left(\Vert x\Vert_2^4-\Vert x\Vert_4^4\right)\le \frac{2}{k}\Vert x\Vert_2^4.
\end{align*}
\end{proof}

\subsubsection{Proof of Claims}
\label{app:omittedproofsSparserclaims}
Throughout this section, we apply Lemma \ref{lem:primitivesxi} repeatedly without further comment.
\asquared*

\begin{proof}
\begin{align*}
    \operatorname{E}_{h}\left[\left(\sum_{i=1}^{k/s}\sum_{r=1}^s\sum_{j=1}^dx_j^2\xi_{ri}(j)\right)^2\right]
    &=\sum_{i,n=1}^{k/s}\sum_{r,t=1}^s\sum_{j,\ell=1}^dx_j^2x_\ell^2\operatorname{E}_{h}\left[\xi_{ri}(j)\xi_{tn}(\ell)\right]\\
    &=\sum_{i,n=1}^{k/s}\sum_{r,t=1}^s\sum_{j=1}^dx_j^4\operatorname{E}_{h}\left[\xi_{ri}(j)\xi_{tn}(j)\right]+\sum_{i,n=1}^{k/s}\sum_{r,t=1}^s\sum_{j\neq\ell=1}^dx_j^2x_\ell^2\operatorname{E}_{h}\left[\xi_{ri}(j)\xi_{tn}(\ell)\right]\\
    &=\sum_{i=1}^{k/s}\sum_{r=1}^s\sum_{j=1}^dx_j^4\frac{s}{k}+\sum_{i=1}^{k/s}\sum_{r\neq t=1}^s\sum_{j=1}^dx_j^4\frac{s^2}{k^2}+\sum_{i\neq n=1}^{k/s}\sum_{r=1}^s\sum_{j=1}^dx_j^4\cdot 0\\&\qquad+\sum_{i\neq n=1}^{k/s}\sum_{r\neq t=1}^s\sum_{j=1}^dx_j^4\frac{s^2}{k^2}+\sum_{i,n=1}^{k/s}\sum_{r,t=1}^s\sum_{j\neq\ell=1}^dx_j^2x_\ell^2\frac{s^2}{k^2}\\
    &=s\Vert x\Vert_4^4+\frac{s(s^2-s)}{k}\Vert x\Vert_4^4+\left(1-\frac{s}{k}\right)(s^2-s)\Vert x\Vert_4^4+s^2\left(\Vert x\Vert_2^4-\Vert x\Vert_4^4\right)\\&=s^2\Vert x\Vert_2^4
\end{align*}
\end{proof}

\bsquared*

\begin{proof}
\begin{align*}
    \operatorname{E}_{h,\sign}\left[\left(\sum_{i=1}^{k/s}\sum_{r=1}^s\sum_{j\neq\ell}x_jx_\ell\sign_r(j)\sign_r(\ell)\xi_{ri}(j)\xi_{ri}(\ell)\right)^2\right]
    &=\sum_{i,n=1}^{k/s}\sum_{r,t=1}^s\sum_{\substack{j\neq\ell\\v\neq w}}x_jx_\ell x_v x_w\operatorname{E}_{\sign}\left[\sign_r(j)\sign_r(\ell)\sign_t(v)\sign_t(w)\right]\\&\qquad\qquad\qquad\qquad\cdot\operatorname{E}_{h}\left[\xi_{ri}(j)\xi_{ri}(\ell)\xi_{tn}(v)\xi_{tn}(w)\right]
\end{align*}
We remark that, as $j\neq \ell$, we have
\begin{align*}
\operatorname{E}_{\sign}\left[\sign_r(j)\sign_r(\ell)\sign_t(v)\sign_t(w)\right]=\begin{cases}
1,\qquad r=t\land((j=v\land \ell=w)\lor(j=w\land \ell=v))\\
0,\qquad \text{otherwise.}
\end{cases}
\end{align*}
This leaves us with
\begin{align*}
2\sum_{i,n=1}^{k/s}\sum_{r=1}^s\sum_{j\neq\ell}x_j^2x_\ell^2 \operatorname{E}_{h}\left[\xi_{ri}(j)\xi_{ri}(\ell)\xi_{rn}(j)\xi_{rn}(\ell)\right]&=2\sum_{i=1}^{k/s}\sum_{r=1}^s\sum_{j\neq\ell}x_j^2x_\ell^2 \operatorname{E}_{h}\left[\xi_{ri}(j)^2\right]\operatorname{E}_{h}\left[\xi_{ri}(\ell)^2\right]\\&=2\sum_{i=1}^{k/s}\sum_{r=1}^s\sum_{j\neq\ell}x_j^2x_\ell^2 \frac{s^2}{k^2}=\frac{2s^2}{k}\left(\Vert x\Vert_2^4-\Vert x\Vert_4^4\right)
\end{align*}
where we used that $\E_{h}[\xi_{ri}(j)\xi_{rn}(j)]=0$ if $i\neq n$.
\end{proof}

\twoab*
\begin{proof}
Observe that 
\begin{align*}
\E_S[ab]
&=\sum_{i,v=1}^{k/s}\sum_{r,t=1}^s\sum_{\substack{j=1\\v\neq w}}^dx_j^2x_{v}x_{w}\operatorname{E}_{\sign}\left[\sign_t(v)\sign_t(w)\right]\operatorname{E}_{h}\left[\xi_{ri}(j)\xi_{tn}(v)\xi_{tn}(w)\right]=0.
\end{align*}
because the signs are independent.
\end{proof}

\end{document}